\documentclass[reqno,10pt,a4paper]{amsart}

\usepackage{amssymb,eucal,mathrsfs,array,setspace,geometry,enumitem,cite}
\usepackage[dvipsnames]{xcolor}
\usepackage[centertableaux]{ytableau}
\usepackage{txfonts}            
\usepackage{mathtools} 
\usepackage{stmaryrd} 
\usepackage[capitalise,noabbrev]{cleveref}

\DeclareSymbolFont{largesymbols}{OMX}{zplm}{m}{n} 

\setitemize{leftmargin=*}   
\setenumerate{leftmargin=*, 
              label=\textup{(\arabic*)}} 

\let\originalleft\left     
\let\originalright\right
\renewcommand{\left}{\mathopen{}\mathclose\bgroup\originalleft}
\renewcommand{\right}{\aftergroup\egroup\originalright}

\newcolumntype{C}{>{$}c<{$}} 

\renewenvironment{quote}{\list{}{\leftmargin=1.5em\rightmargin=1.5em}\item[]}{\endlist}

\numberwithin{equation}{section}
\allowdisplaybreaks

\newcommand{\sfrac}[2]{#1/#2}

\renewcommand{\ge}{\geq}
\renewcommand{\le}{\leq}

\DeclarePairedDelimiter{\brac}{\lparen}{\rparen} 
\DeclarePairedDelimiter{\sqbrac}{\lbrack}{\rbrack} 
\DeclarePairedDelimiter{\set}{\lbrace}{\rbrace}
\newcommand{\st}{\mspace{5mu} : \mspace{5mu}} 
\DeclarePairedDelimiter{\abs}{\lvert}{\rvert}

\DeclarePairedDelimiter{\ang}{\langle}{\rangle}
\DeclarePairedDelimiter{\normord}{{} :}{: {}} 
\DeclarePairedDelimiter{\powser}{\llbracket}{\rrbracket} 

\DeclarePairedDelimiterX{\comm}[2]{\lbrack}{\rbrack}{#1 , #2}  
\DeclarePairedDelimiterX{\acomm}[2]{\lbrace}{\rbrace}{#1 , #2} 
\DeclarePairedDelimiterX{\super}[2]{\lparen}{\rparen}{#1 \delimsize\vert \mathopen{} #2} 

\DeclareMathOperator{\pf}{pf}
\DeclareMathOperator{\sgn}{sgn}
\DeclareMathOperator{\id}{id}

\newcommand{\pd}{\partial}     

\newcommand{\dd}{\mathrm{d}}   
\newcommand{\ee}{\mathsf{e}}   

\newcommand{\wun}{\mathbf{1}}  

\DeclareMathOperator{\Res}{Res}
\DeclareMathOperator{\spn}{span}
\newcommand{\res}[1]{\Res_{#1}}                            

\newcommand{\ira}{\hookrightarrow}    

\DeclareMathOperator{\ind}{Ind}
\newcommand{\Ind}[3]{\ind^{#1}_{#2} #3}

\newcommand{\fld}[1]{\mathbb{#1}}    
\newcommand{\alg}[1]{\mathfrak{#1}}  
\newcommand{\grp}[1]{\mathsf{#1}}    
\newcommand{\Mod}[1]{\mathcal{#1}}   
\newcommand{\VOA}[1]{\mathsf{#1}}    

\newcommand{\ZZ}{\fld{Z}}

\newcommand{\QQ}{\fld{Q}}
\newcommand{\RR}{\fld{R}}
\newcommand{\CC}{\fld{C}}

\newcommand{\symgp}[1]{\grp{S}_{#1}}   

\newcommand{\affine}[1]{\widehat{#1}}
\newcommand{\AKMA}[2]{\affine{\alg{#1}} \brac[\big]{#2}}       
\newcommand{\AKMSA}[3]{\affine{\alg{#1}} \super[\big]{#2}{#3}} 

\newcommand{\halg}{\alg{h}}                                                 
\newcommand{\falg}[1]{\alg{f}_{#1}}                                         
\newcommand{\svir}[1]{\alg{svir}_{#1}}                                      

\newcommand{\MinMod}[2]{\VOA{M}\left( #1 , #2 \right)}                      
\newcommand{\svc}[1]{\VOA{V}(#1)}                                           
\newcommand{\hvoa}[1]{\VOA{H}(#1)}                                          
\newcommand{\fvosa}{\VOA{F}}                                                
\newcommand{\ff}{\VOA{HF}}                                                  
\newcommand{\ffvoa}[1]{\ff(#1)}                                             

\newcommand{\ideal}[1]{\ang[\big]{#1}}
\newcommand{\vsv}[1]{\chi_{#1}}                                             

\newcommand{\gtzhu}[2]{\mathsf{A}^{\!#2}\sqbrac*{#1}}                       
\newcommand{\zhu}[1]{\gtzhu{#1}{}}                                          
\newcommand{\parity}{\tau}                                                  
\newcommand{\tzhu}[1]{\gtzhu{#1}{\parity}}                                  

\newcommand{\gtozhu}[2]{\mathsf{O}^{#2}\sqbrac*{#1}}                        
\newcommand{\ozhu}[1]{\gtozhu{#1}{}}
\newcommand{\tozhu}[1]{\gtozhu{#1}{\parity}}

\DeclareMathOperator{\wt}{wt}                                               
\newcommand{\zstar}{\ast}                                                   
\newcommand{\zcirc}{\circ}   

\newcommand{\zmsub}[2]{\VOA{#1}_{#2}}                                       
\newcommand{\tzmsub}[2]{\VOA{#1}^{\parity}_{#2}}                            

\newcommand{\NS}{\mathrm{NS}}
\newcommand{\R}{\mathrm{R}}
\newcommand{\N}{\mathrm{N}}

\newcommand{\Ver}{\Mod{M}}                   
\newcommand{\Irr}{\Mod{L}}                   
\newcommand{\NSVer}[1]{\Ver^{\NS}(#1)}       
\newcommand{\NSIrr}[1]{\Irr^{\NS}(#1)}       
\newcommand{\RVer}[1]{\Ver^{\R}(#1)}         
\newcommand{\RIrr}[1]{\Irr^{\R}(#1)}         

\newcommand{\Fock}[1]{\Mod{F}_{#1}}                       
\newcommand{\NSFock}{\Fock{}^{\NS}}          
\newcommand{\RFock}{\Fock{}^{\R}}            
\newcommand{\bNSFock}[1]{\mathbb{F}_{#1}^{\NS}}
\newcommand{\bRFock}[1]{\mathbb{F}_{#1}^{\R}}

\DeclarePairedDelimiter{\bra}{\langle}{\rvert}
\DeclarePairedDelimiter{\ket}{\lvert}{\rangle}
\DeclarePairedDelimiterX{\braket}[2]{\langle}{\rangle}{#1 \delimsize\vert \mathopen{} #2}
\DeclarePairedDelimiterX{\bracket}[3]{\langle}{\rangle}{#1 \delimsize\vert \mathopen{} #2 \delimsize\vert \mathopen{} #3}
\newcommand{\brab}[1]{\bra[\big]{#1}}
\newcommand{\ketb}[1]{\ket[\big]{#1}}
\newcommand{\braketb}[2]{\braket[\big]{#1}{#2}}
\newcommand{\bracketb}[3]{\bracket[\big]{#1}{#2}{#3}}
\newcommand{\NSbra}{\brab{\NS}}                  
\newcommand{\NSffbra}[1]{\brab{#1;\NS}}          
\newcommand{\Rbra}{\brab{\R}}                    
\newcommand{\Rffbra}[1]{\brab{#1;\R}}            
\newcommand{\NSket}{\ketb{\NS}}                  
\newcommand{\NSffket}[1]{\ketb{#1;\NS}}          
\newcommand{\Rket}{\ketb{\R}}                    
\newcommand{\Rffket}[1]{\ketb{#1;\R}}            
\newcommand{\NSbracket}[1]{\bracketb{\NS}{#1}{\NS}}
\newcommand{\NSffbracket}[3]{\bracketb{#1;\NS}{#2}{#3;\NS}}

\newcommand{\Rffbracket}[3]{\bracketb{#1;\R}{#2}{#3;\R}}
\DeclarePairedDelimiter{\corrfn}{\langle}{\rangle}   
\newcommand{\corrfnb}[1]{\corrfn[\big]{#1}}
\newcommand{\NScorrfn}[1]{\corrfnb{#1}_{\NS}}        
\newcommand{\Rcorrfn}[1]{\corrfnb{#1}_{\R}}          

\newcommand{\kacsymbol}{\mathsf{K}}
\newcommand{\kac}[1]{\kacsymbol_{#1}}                
\newcommand{\NSkac}[1]{\kac{#1}^{\NS}}
\newcommand{\Rkac}[1]{\kac{#1}^{\R}}
\newcommand{\rkac}[1]{\overline{\kacsymbol}_{#1}}    
\newcommand{\rNSkac}[1]{\rkac{#1}^{\NS}}
\newcommand{\rRkac}[1]{\rkac{#1}^{\R}}

\newcommand{\parl}[1]{\pi_{#1}}                   
\newcommand{\admp}[2]{\delta^{(#1)}\brac*{#2}}      
\newcommand{\uniqp}[2]{\epsilon^{(#1)}\brac*{#2}}   

\newcommand{\vop}[2]{\mathrm{V}_{#1}\brac*{#2}}
\newcommand{\SCR}{\mathcal{Q}}

\newcommand{\scrf}[2]{\SCR_{#1}(#2)}

\newcommand{\scrs}[2]{\SCR_{#1}^{[#2]}}

\newcommand{\cyc}[2]{\Gamma(#1,#2)}

\newcommand{\Van}{\Delta}
\newcommand{\van}[1]{\Van\brac*{#1}}

\newcommand{\sym}{\Lambda}                                       
\newcommand{\fsym}[1]{\sym_{#1}}                                 

\newcommand{\poly}[1]{\mathsf{#1}}
\newcommand{\monsym}[1]{\poly{m}_{#1}}                           
\newcommand{\fmonsym}[2]{\poly{m}_{#1} \brac[\big]{#2}}          

\newcommand{\powsum}[1]{\poly{p}_{#1}}                           

\newcommand{\fpowsum}[2]{\poly{p}_{#1} \brac[\big]{#2}}

\newcommand{\jack}[2]{\poly{P}_{#1}^{#2}}                        
\newcommand{\fjack}[3]{\poly{P}_{#1}^{#2} \brac[\big]{#3}}
\newcommand{\djack}[2]{\poly{Q}_{#1}^{#2}}                       
\newcommand{\fdjack}[3]{\poly{Q}_{#1}^{#2} \brac[\big]{#3}}

\newcommand{\jprod}[3]{\ang*{#1}_{#2}^{#3}}    
\newcommand{\cjprod}[2]{\ang*{#1}^{#2}}        

\newcommand{\JdJ}[3]{\mathcal{N}_{#1}(#2; #3)} 

\newcommand{\symiso}[2]{\rho^{#1}_{#2}}                  

\newcommand{\uea}{universal enveloping algebra}

\newcommand{\uesa}{universal enveloping superalgebra}

\newcommand{\lw}{lowest-weight}
\newcommand{\lwv}{\lw{} vector}
\newcommand{\lwvs}{\lw{} vectors}

\newcommand{\hw}{highest-weight}
\newcommand{\hwv}{\hw{} vector}
\newcommand{\hwvs}{\hw{} vectors}
\newcommand{\hwm}{\hw{} module}
\newcommand{\hwms}{\hw{} modules}
\newcommand{\sv}{singular vector}
\newcommand{\svs}{singular vectors}
\newcommand{\voa}{vertex operator algebra}
\newcommand{\voas}{vertex operator algebras}
\newcommand{\vosa}{vertex operator superalgebra}
\newcommand{\vosas}{vertex operator superalgebras}
\newcommand{\ope}{operator product expansion}

\newcommand{\opes}{operator product expansions}

\newcommand{\PBW}{Poincar\'{e}-Birkhoff-Witt}
\newcommand{\ns}{Neveu-Schwarz}
\newcommand{\lhs}{left-hand side}
\newcommand{\rhs}{right-hand side}
\newcommand{\lhss}{left-hand sides}

\theoremstyle{plain}
\newtheorem{thm}{Theorem}[section]
\newtheorem{prop}[thm]{Proposition}
\newtheorem{lem}[thm]{Lemma}
\newtheorem{cor}[thm]{Corollary}

\newtheorem{defn}[thm]{Definition}
\newtheorem*{rmk}{Remark}
\newtheorem*{thm*}{Theorem}

\Crefname{thm}{Theorem}{Theorems}
\Crefname{prop}{Proposition}{Propositions}
\Crefname{lem}{Lemma}{Lemmas}
\Crefname{cor}{Corollary}{Corollaries}
\Crefname{defn}{Definition}{Definitions}

\begin{document}

\title[]{Superconformal minimal models and admissible Jack polynomials}

\author[O Blondeau-Fournier]{Olivier Blondeau-Fournier}

\address[Olivier Blondeau-Fournier]{
Department of Mathematics \\
King's College London \\
Strand, United Kingdom, WC2R~2LS.
}

\email{olivier.blondeau-fournier@kcl.ac.uk}

\author[P Mathieu]{Pierre Mathieu}

\address[Pierre Mathieu]{
D\'epartement de Physique, de G\'enie Physique et d'Optique \\
Universit\'e Laval \\ 
Qu\'ebec, Canada, G1V~0A6.
}

\email{pmathieu@phy.ulaval.ca}

\author[D Ridout]{David Ridout}

\address[David Ridout]{
School of Mathematics and Statistics \\
University of Melbourne \\
Parkville, Australia, 3010.
}

\email{david.ridout@unimelb.edu.au}

\author[S Wood]{Simon Wood}

\address[Simon Wood]{
Mathematical Sciences Institute \\
Australian National University \\
Acton, Australia, 2601.
}

\email{woodsi@cardiff.ac.uk}

\thanks{\today}

\begin{abstract}
  We give new proofs of the rationality of the \(N=1\) superconformal minimal model \vosas{} and of the classification of their modules in both the Neveu-Schwarz and Ramond sectors. For this, we combine the standard free field realisation with the theory of Jack symmetric functions. A key role is played by Jack symmetric polynomials with a certain negative parameter that are labelled by admissible partitions. These polynomials are shown to describe free fermion correlators, suitably dressed by a symmetrising factor. The classification proofs concentrate on explicitly identifying Zhu's algebra and its twisted analogue. Interestingly, these identifications do not use an explicit expression for the non-trivial vacuum singular vector.  While the latter is known to be expressible in terms of an Uglov symmetric polynomial or a linear combination of Jack superpolynomials, it turns out that standard Jack polynomials (and functions) suffice to prove the classification.
\end{abstract}

\maketitle

\onehalfspacing

\section{Introduction} \label{sec:Intro}

The purpose of this article is to give a new proof of the classification of the simple modules of the \(N=1\) superconformal minimal model \vosas{} \(\MinMod{p_+}{p_-}\) in the Neveu-Schwarz and Ramond sectors.  The rationality in both sectors is also established. The proof of this classification makes use of a deep connection between the theory of symmetric functions and free field realisations.  Moreover, the method of proof in both sectors is essentially the same.  This method has previously been applied to classify the simple modules of the Virasoro minimal models \cite{RidJac14}, the admissible level affine \(\AKMA{sl}{2}\) models \cite{RidSlJac15} and the triplet algebras \cite{TsuExt13}.

Let $p_+$ and $p_-$ be integers satisfying $p_+, p_- \ge 2$, $p_- - p_+ \in 2 \ZZ$ and $\gcd \set{\frac{1}{2} (p_- - p_+), p_-} = 1$.  Let
\begin{equation}
c_{p_+,p_-} = \frac{3}{2}-3\frac{(p_--p_+)^2}{p_+p_-}, \quad
h_{r,s}=\frac{(rp_--sp_+)^2-(p_--p_+)^2}{8p_+p_-}+\frac{1-(-1)^{r+s}}{32},
\end{equation}
where $r$ and $s$ are positive integers.  Additionally, let $\NSIrr{h,c}$ and $\RIrr{h,c}$ denote the simple \hwms{} over the \ns{} and Ramond algebras, respectively, whose \hwvs{} have conformal weight $h$, central charge $c$ and even parity.  Then, we can state the main result as follows (referring to \cref{sec:N=1VOSA} for our conventions concerning modules and the notion of parity reversal).
\begin{thm*}
  The $N=1$ superconformal minimal model \vosa{} \(\MinMod{p_+}{p_-}\) is rational in both the Neveu-Schwarz and Ramond sectors, that is,
  both sectors have finitely many simple \(\ZZ_2\)-graded modules and every \(\ZZ_2\)-graded module is semisimple.
  \begin{enumerate}
  \item Up to isomorphism, the simple $\MinMod{p_+}{p_-}$-modules in the Neveu-Schwarz sector are given
  by the \(\NSIrr{{h_{r,s},c_{p_+,p_-}}}\), with $1 \le r \le p_+-1$, $1 \le s \le p_--1$ and $r+s \in 2 \ZZ$,
  and their parity reversals.
  \item Up to isomorphism, the simple $\MinMod{p_+}{p_-}$-modules in the Ramond sector are given by the
  \(\RIrr{h_{r,s},c_{p_+,p_-}}\), with \mbox{$1 \le r \le p_+-1$}, $1 \le s \le p_--1$ and $r+s \in 2 \ZZ + 1$,
  and, if $p_+$ is even, the parity reversal of \(\RIrr{h_{p_+/2,p_-/2}}\)
  (the other simple Ramond modules being isomorphic to their parity-reversed counterparts).
\end{enumerate}
\end{thm*}

The (non-rigorous) classification of the simple modules appearing in the $N=1$ minimal models was, of course, very well known to physicists \cite{EicMin85,BerSup85,FriSup85} and the celebrated coset construction confirmed their results for the unitary minimal models $\MinMod{p_+}{p_++2}$ \cite{GodUni86}.  However, rigorous proofs that included the non-unitary models remained elusive.  Following Wang's explicit identification of Zhu's algebra for the Virasoro minimal models \cite{WanRat93}, Kac and Wang conjectured the corresponding result for the $N=1$ minimal models \cite{Kacn1z94}, but were unable to provide a proof for the non-unitary cases.  Subsequently, Adamovi\'{c} \cite{AdaRat97} extended the coset proof to the non-unitary cases as a simple consequence of his classification \cite{AdaVer95}, obtained with Milas, of the simple modules of the admissible level $\AKMA{sl}{2}$ models.  However, he only determined which $N=1$ modules appeared in the \ns{} sector.  The coset construction also produces the simple modules in the Ramond sector, but they were not considered because Zhu's algebra cannot be used to determine whether one has indeed found them all.  The appropriate generalisation of Zhu's algebra appeared shortly thereafter \cite{DonTwi98}, but it seems that a complete proof for the Ramond classification did not appear until \cite{MiltWeb07}.

Our classification proof applies to both \ns{} and Ramond sectors and is not contingent on a coset construction.  As noted above, it instead relies on embedding the $N=1$ \vosa{} into a free field \vosa{} and using tools from the theory of symmetric polynomials to calculate within the latter.  This connection between symmetric polynomials and free field realisations originated in the work of Wakimoto and Yamada \cite{WakFoc86b} and was continued in \cite{KatMis92,MimSin95,AwaWN95}, where it was used to derive compact formulae for singular vectors of various \voas{} in terms of their free field realisations. However, the actual utilisation of these \sv{} formulae for classifying irreducible modules appears to be new.

Recently, there has been a resurgence of interest in using symmetric polynomials to construct singular vectors, particularly for the \(N=1\) superconformal \vosas{}, thanks to the AGT conjectures \cite{AldLio10}.  In particular, there have been two parallel developments that are closely related to the work reported here. One approach \cite{BelUg13,YanUg15} uses a basis of symmetric polynomials called Uglov polynomials \cite{UglYan98}, a specialisation of Macdonald polynomials that are similar to Jack polynomials, and leads to \sv{} formulae involving a single Uglov polynomial.  However, this has thus far only been studied in the \ns{} sector.  The other approach \cite{DesSup01} instead works with superspace analogues of Jack polynomials, called Jack superpolynomials, that directly incorporate anticommuting (Grassmann) variables.  Singular vector formulae have been conjectured in both the \ns{} and Ramond sectors \cite{DesSJa12,AlaRam13} and similar results have recently been rigorously proved \cite{OPDS}.  However, these formulae involve linear combinations of Jack superpolynomials.

Our work differs from these approaches in that we are not interested in explicit \sv{} formulae themselves.  Rather, the point is to instead use implicit formulae for \svs{} to explicitly identify Zhu's algebras for the $N=1$ superconformal minimal models and thereby classify the irreducible representations in the \ns{} and Ramond sectors.  A simple corollary of this is the rationality of these minimal models.  Moreover, we do not employ Uglov polynomials nor Jack superpolynomials in proving the classification theorem, but instead find that the standard Jack symmetric polynomials are sufficient.  This does require some more sophisticated tools.  In particular, our proofs rely on the theory of negative parameter Jack polynomials associated to admissible partitions that was introduced by Feigin, Jimbo, Miwa and Mukhin \cite{Feidif02}.  This aside, many of the arguments are still significantly more involved than one would expect given the elegance of the arguments for the (non-super) Virasoro minimal models \cite{RidJac14}.  It will be very interesting to determine whether our pure-Jack formalism can be generalised to accommodate Uglov and/or Jack superpolynomials and thereby recover this expected elegance.  We mention that the recent results of \cite{OPDS} show that the non-trivial singular vector in the vacuum module can be expressed in terms of Jack superpolynomials.  However, the calculations that connect this expression to Zhu's algebras turn out to be independent of the superspace construction (the anticommuting variables) and reduce to those reported here.

\medskip

This article is organised as follows. \cref{sec:N=1} begins with a review of the \(N=1\) universal \vosas{} and their simple quotients, the \(N=1\) superconformal minimal model \vosas{}.  This is followed by a description of their standard free field realisations and an outline of the construction of screening operators, essential for the \sv{} computations to come.  The section concludes with derivations of explicit formulae for certain correlation functions, particularly those involving free fermions.  Most of this material is standard, but is included for completeness as well as to fix notation and conventions.

The main topic of \cref{sec:Jack} is an important ideal of the ring of symmetric polynomials that is intimately connected to Jack polynomials that are labelled by a given negative parameter and the so-called admissible partitions.  This is actually a special case of a much more general picture that was introduced and studied in \cite{Feidif02}.  We begin by collecting a few combinatorial results concerning admissible partitions that will be used in the calculations that follow.  The main goal is to express the free fermion correlation functions of the previous section in terms of Jack polynomials for certain admissible partitions.  The results are very elegant for the \ns{} correlators, but their Ramond analogues are (perhaps unsurprisingly) somewhat more complicated.

\cref{sec:zhu} then combines these expressions for the fermion correlators with the symmetric polynomial theory detailed in \cite{MacSym95} to identify Zhu's algebra and its twisted generalisation for any $N=1$ superconformal minimal model.  These identifications quickly yield the desired classification and rationality of the corresponding \vosas{} in the \ns{} and Ramond sectors, respectively.  Generalising the point of view of \cite[App.~B]{RidSlJac15}, we explain in \cref{sec:twistedZhu} that the definition of twisted and untwisted Zhu algebras is nothing but an abstraction of the action of zero modes on ground states.  We also emphasise that a field only induces an element of a given Zhu algebra if it has a zero mode when acting in the corresponding sector.  It seems that this point of view is rarely made explicit in the literature.  In our opinion, this greatly obscures the underlying simplicity and utility of Zhu theory.

The actual calculation of the twisted and untwisted Zhu algebras for the $N=1$ minimal models first notes that these algebras are quotients of polynomial rings in a single variable.  The goal therefore reduces to computing a single polynomial for each.  These polynomials may, in turn, be determined by studying which \hwvs{} are annihilated by the zero mode of a single (carefully chosen) null field.  Our first result is that this null field may be constructed in the free field realisation.  The proof uses the Jack polynomial technology developed in \cref{sec:Jack}.  Our second result is that the corresponding polynomials are in fact non-zero.  This follows in the untwisted case from a quite general argument, but the twisted version of this is considerably more involved and is instead proven as a corollary to the identification of the untwisted polynomial.  These results then allow us to attend to our main result, the actual identification of these polynomials (which also requires the free field realisation and Jack technology).  The calculations are notable for the fact that the methodology does not appear to allow these polynomials to be computed directly, unlike the cases detailed in \cite{RidJac14,RidSlJac15}.  Nevertheless, we are able to determine sufficiently many zeroes that complete identifications can be made by appealing to an obvious symmetry property.  It would be very interesting to determine whether these polynomials may be directly determined by generalising to Uglov or Jack superpolynomials.

\section*{Acknowledgements}

OBF is supported by le Fonds de Recherche du Qu\'{e}bec --- Nature et Technologies.  PM's research is supported by the Natural Sciences and Engineering Research Council of Canada.  DR's research is supported by the Australian Research Council Discovery Projects DP1093910 and DP160101520.  SW is supported by the Australian Research Council Discovery Early Career Researcher Award DE140101825 and the Discovery Project DP160101520.

\section{\(N=1\) superalgebras and their correlation functions} \label{sec:N=1}

In this section, we recall several well known results concerning the \(N=1\) \vosas{} and their free field realisations.  This also serves to settle notation and conventions for the sections that follow.

\subsection{\(N=1\) \vosas{}} \label{sec:N=1VOSA}

The $N=1$ superconformal algebras are a pair of infinite-dimensional complex Lie superalgebras parametrised by a label $\epsilon \in \set{0, \frac{1}{2}}$:
\begin{equation}
  \svir{\epsilon}=\bigoplus_{n\in\ZZ}\CC
  L_n\oplus\:\bigoplus_{\mathclap{m\in\ZZ+\epsilon}}\:\CC G_m\oplus \CC C.
\end{equation}
This defines a vector space direct sum decomposition into an even (bosonic) subspace, spanned by the $L_n$ and $C$, and an odd (fermionic) subspace, spanned by the $G_m$.  The superalgebra with $\epsilon = \frac{1}{2}$ is known as the \emph{\ns{} algebra} \cite{NevFac71} and that with $\epsilon = 0$ is the \emph{Ramond algebra} \cite{RamDua71}.  The defining Lie brackets of both are given by
\begin{align} \label{eq:N=1CommRels}
\begin{aligned}
  \comm{L_m}{L_n}&=(m-n)L_{m+n}+\frac{1}{12}(m^3-m)\delta_{m+n,0}C,\\
  \comm{L_m}{G_r}&=\left(\frac12m-r\right)G_{m+r},\\
  \acomm{G_r}{G_s}&=2L_{r+s}+\frac{1}{3}\left(r^2-\frac14\right)\delta_{r+s,0}C,
\end{aligned}
\qquad
\begin{aligned}
&m,n\in\ZZ, \\
&r,s\in\ZZ+\epsilon,
\end{aligned}
\end{align}
and \(C\) is central. We identify \(C\) with a multiple of the identity, \(C=c\cdot\id\), when acting on modules and refer to the number \(c\in\CC\) as the central charge. Modules over the \ns{} algebra are said to belong to the \ns{} sector, while modules over the Ramond algebra belong to the Ramond sector.

For reasons coming from physics (which are discussed at the end of this subsection), we shall require that all superalgebra modules are $\ZZ_2$-graded, meaning that they admit a vector space direct sum decomposition into an even and an odd subspace.  This decomposition must be compatible with that of the superalgebra so that the action of an even superalgebra element preserves the even and odd subspaces of the module, while the action of an odd element maps between these two subspaces.  It follows that there is an ambiguity in imposing this structure on a given indecomposable module over the superalgebra, even once the vector space decomposition has been agreed upon, because we may swap the even and odd subspaces with impunity.  In general, each indecomposable superalgebra module therefore comes in two flavours, isomorphic as modules but not as $\ZZ_2$-graded modules, which only differ in the global choice of parity.  Given a superalgebra module, we shall refer to the module obtained by swapping its even and odd subspaces as its \emph{parity reversal}.  Of course, it may happen that a module and its parity reversal are isomorphic as $\ZZ_2$-graded modules.

Recall the standard triangular decomposition of the \ns{} algebra:
\begin{align}
 \svir{\sfrac{1}{2}}^\pm=\bigoplus_{n>0}\CC L_{\pm n} \oplus 
 \bigoplus_{m>0}\CC G_{\pm m},
 \quad
  \svir{\sfrac{1}{2}}^0=\CC L_0\oplus\CC C.
\end{align}
Writing \(\svir{\sfrac{1}{2}}^{\ge}=\svir{\sfrac{1}{2}}^+\oplus\svir{\sfrac{1}{2}}^0\), the \ns{} Verma module
\begin{equation}\label{eq:NSverm}
  \NSVer{h,c}=\Ind{\svir{\sfrac{1}{2}}}{\svir{\sfrac{1}{2}}^{\ge}}{\N(h,c)}
\end{equation}
is induced from the $1$-dimensional \(\svir{\sfrac{1}{2}}^\ge\)-module $\N(h,c)=\CC \Omega_{h,c}$ characterised by the parity of the generating vector $\Omega_{h,c}$ being even and
\begin{equation}
  L_0\Omega_{h,c}=h\Omega_{h,c},\quad C\Omega_{h,c}=c\Omega_{h,c}
  ,\quad \svir{\sfrac{1}{2}}^+\ \Omega_{h,c}=0.
\end{equation}
There are, in addition, parity-reversed \ns{} Verma modules that are induced from an odd vector.  Note that \ns{} Verma modules are never isomorphic (as $\ZZ_2$-graded modules) to their parity-reversed counterparts.
By the structure theory of these modules \cite{AstStr97}, this also holds true for \ns{} \hwms{}.

The construction of Ramond Verma modules is slightly different as the decomposition
\begin{align}
 \svir{0}^\pm = \bigoplus_{n>0} \CC L_{\pm n} \oplus \bigoplus_{m>0} \CC G_{\pm m},
 \quad
 \svir{0}^0 = \CC L_0 \oplus \CC G_0 \oplus \CC C
\end{align}
is not a triangular decomposition of the Ramond algebra, because $\svir{0}^0$ is not abelian.  However, we may proceed instead via \emph{generalised} Verma modules which are induced from an arbitrary simple $\svir{0}^0$-module.  The following classification follows easily from the fact that $G_0^2 = L_0 - \frac{1}{24} C$ (in the \uea{}).
\begin{prop}
The finite-dimensional, $\ZZ_2$-graded, simple $\svir{0}^0$-modules are classified by the unique eigenvalues $h$ and $c$ of $L_0$ and $C$, respectively, and the global parity.
\begin{itemize}
\item If $h \neq \frac{c}{24}$, then there is exactly one such module (up to isomorphism), denoted by $\R(h,c)$.  Its dimension is $2$ and it is isomorphic to its parity reversal.
\item If $h = \frac{c}{24}$, then there are exactly two such modules (up to isomorphism):  $\R(c/24,c)$ and its parity reversal.  Their dimensions are $1$.
\end{itemize}
\end{prop}
\noindent For each $h,c \in \CC$ (and each choice of global parity), we may extend $\R(h,c)$ to an $\svir{0}^{\ge}$-module and then induce to obtain the Ramond Verma module
\begin{equation}\label{eq:Rverm}
  \RVer{h,c}=\Ind{\svir{0}}{\svir{0}^{\ge}}{\R(h,c)}.
\end{equation}
Ramond Verma modules with $h \neq \frac{c}{24}$ are always isomorphic to their parity-reversed counterparts, while those with $h = \frac{c}{24}$ never are. Again, this statement also holds for Ramond \hwms{} \cite{IohVer03}.

For \(h=0\), the \ns{} Verma module \(\NSVer{0,c}\) is reducible and the singular vector \(G_{-\sfrac{1}{2}}\Omega_{0,c}\) generates a proper submodule. We denote the quotient by
\begin{equation}
  \svc{c}=\frac{\NSVer{0,c}}{\ideal{G_{-\sfrac{1}{2}}\Omega_{0,c}}}.
\end{equation}
It carries the structure of an $N=1$ \vosa{}.
\begin{defn}
  The \emph{universal \(N=1\) \vosas{}} are the unique \vosas{} that are strongly generated by an even field \(T(z)\) and an odd field \(G(z)\), have defining
  \opes{}
  \begin{align}\label{eq:n1ope}
  \begin{aligned}
    T(z)T(w)&\sim \frac{c/2}{(z-w)^4}+\frac{2T(w)}{(z-w)^2}+\frac{\pd T(w)}{(z-w)},\\
    T(z)G(w)&\sim\frac{\frac32 G(w)}{(z-w)^2}+\frac{\pd G(w)}{z-w},\\ 
    G(z)G(w)&\sim \frac{2c/3}{(z-w)^3}+\frac{2 T(w)}{z-w},
  \end{aligned}
  \end{align}
  and satisfy
  no additional relations beyond
  those required by the \vosa{} axioms. These \vosas{} are parametrised by the central charge $c \in \CC$.
\end{defn}
\noindent We recall that a \vosa{} is strongly generated by a set of fields if any field of the \vosa{} may be written as a normally ordered polynomial in the fields of the generating set and their derivatives. In particular, since the fields \(T(z)\) and \(G(z)\) of the universal \(N=1\) \vosa{} \(\svc{c}\) satisfy no relations other than the \opes{} \eqref{eq:n1ope}, the set of all normally ordered monomials of derivatives of \(T(z)\) and \(G(z)\) form a basis of \(\svc{c}\) (after also imposing a \PBW{} ordering on the monomials).

The \opes{} \eqref{eq:n1ope} imply that the modes of the Laurent expansions
\begin{equation}
T(z)=\sum_{n\in\ZZ}L_nz^{-n-2}, \quad G(z)=\:\sum_{\mathclap{n\in\ZZ+\sfrac{1}{2}}}\:G_n z^{-n-\sfrac{3}{2}}
\end{equation}
satisfy the commutation relations \eqref{eq:N=1CommRels} of \(\svir{\sfrac{1}{2}}\). Indeed, 
as an \(\svir{\sfrac{1}{2}}\)-module, this universal \vosa{} is isomorphic to \(\svc{c}\) and so we will denote it by the same symbol.

\begin{prop}[Astashkevich \cite{AstStr97}] \label{prop:N=1Simple}
  The universal \(N=1\) \vosa{} \(\svc{c}\) contains a proper non-trivial ideal if and only if
  \begin{equation}
    c=c_{p_+,p_-}=\frac32-3\frac{(p_--p_+)^2}{p_+p_-},
  \end{equation}
  for some positive integers $p_+$ and $p_-$ satisfying \(p_->p_+\ge2\), \(p_--p_+\in 2\ZZ\) and \(\gcd\set{\frac{1}{2}(p_--p_+),p_-}=1\). For these central charges, 
  the maximal proper ideal is simple and it is generated by a singular vector \(\vsv{p_+,p_-}\) of conformal weight \(\frac12(p_+-1)(p_--1)\).
\end{prop}
\noindent Note that the ordering \(p_->p_+\) is not required, but we shall assume it for later convenience.
\begin{defn}
  For $p_+$ and $p_-$ satisfying \(p_->p_+\ge2\), \(p_--p_+\in 2\ZZ\) and \(\gcd\set{\frac{1}{2}(p_--p_+),p_-}=1\), the minimal model $N=1$ \vosa{} $\MinMod{p_+}{p_-}$ is defined to be the unique simple quotient of the universal $N=1$ \vosa{} \(\svc{c_{p_+,p_-}}\) by its maximal proper ideal:
  \begin{equation}
    \MinMod{p_+}{p_-}=\frac{\svc{c_{p_+,p_-}}}{\ideal{\vsv{p_+,p_-}}}.
  \end{equation}
\end{defn}

We conclude by formalising the type of modules (and twisted modules) that we wish to classify.  As noted above, we insist that all modules possess a $\ZZ_2$-grading that is consistent with that of the superalgebra.  This is required for many conformal field-theoretic applications including those that require fusion or modular transformations (supercharacters).  Indeed, for (super)characters to exist, we must also require the finite-dimensionality of the generalised eigenspaces of $L_0$ (this, in turn, implies that any Jordan blocks for $L_0$ have finite rank).  We therefore make the following declaration:
\begin{quote}
	\emph{Unless otherwise noted, all (twisted) modules $\Mod{M}$ over a \vosa{} shall be understood to be
	\begin{itemize}
		\item finitely generated;
		\item $\ZZ_2$-graded, in the sense described above;
		\item a direct sum $\Mod{M} = \bigoplus_{n \ge h} \Mod{M}_n$ of finite-dimensional generalised $L_0$-eigenspaces $\Mod{M}_n$ of eigenvalue $n$.
	\end{itemize}
	In what follows, we shall generally only qualify modules explicitly as
        being $\ZZ_2$-graded, for brevity.}
\end{quote}
We emphasise the imposition of the lower bound $h$ on the eigenvalues of $L_0$
on $\Mod{M}$.  This serves to guarantee that a non-zero module will possess
vectors of minimal $L_0$-eigenvalue and thus will yield a non-zero module over
the corresponding Zhu algebras (see \cref{sec:twistedZhu}). The lower bound on
the conformal weight also guarantees that the action of a field on a module element is a formal Laurent series, compatible with the requirements of \opes{}.

\subsection{The free field realisation}

In this section we define the free boson and free fermion \vosas{} and embed the universal $N=1$ \vosa{} 
\(\svc{c}\) into their tensor product. Such an embedding is called a \emph{free field realisation} of $\svc{c}$.

\subsubsection{The Heisenberg algebra \(\halg\)} \label{sec:FreeBoson}

The \emph{Heisenberg algebra} is the infinite-dimensional complex Lie algebra
\begin{equation}
  \halg=\bigoplus_{n\in\ZZ}\CC a_n\oplus\CC\wun,
\end{equation}
whose Lie brackets are
\begin{equation}\label{eq:heiscomrel}
  [a_m,a_n]=m\delta_{m+n,0}\wun, \quad m,n \in \ZZ.
\end{equation}
The element \(\wun\) is central and will always be identified 
with the identity when acting on an $\halg$-module.\footnote{Provided that the central element acts non-trivially, the   generators \(a_n\) can always be rescaled so that
the central element acts as the identity on a simple module.}

The Heisenberg algebra admits the triangular decomposition
\begin{equation}
  \halg^\pm= \bigoplus_{n>0}\CC a_{\pm n},\quad \halg^0=\CC a_0\oplus\CC \wun,
\end{equation}
which we shall use to construct
Verma modules.  Writing
\(\halg^\ge=\halg^+\oplus\halg^0\) as usual, we define 
\begin{equation}
  \Fock{p}=\Ind{\halg}{\halg^\ge}{\CC \ketb{p}},\quad p\in\CC,
\end{equation}
to be the Verma module induced from the $1$-dimensional \(\halg^{\ge}\)-module characterised by
\begin{equation}
  a_0\ketb{p}=p\ketb{p},\quad\wun\ketb{p}=\ketb{p},\quad \halg^+\ketb{p}=0.
\end{equation}
Verma modules for the Heisenberg algebra are always simple and are also known as \emph{Fock spaces}.

\begin{defn}
  The \emph{Heisenberg \voas{}} \(\hvoa{\alpha_0}\), also known as the \emph{(deformed) free boson}, are the unique \voas{} that are strongly generated by a field \(a(z)\), have defining \ope{}
  \begin{equation}\label{eq:hope}
    a(z)a(w)\sim\frac{1}{(z-w)^2},
  \end{equation}
  and satisfy no additional relations beyond those required by the \voa{} axioms.  These \voas{} are parametrised by $\alpha_0\in\CC$ which determines the choice of energy-momentum tensor:
  \begin{equation} \label{eq:HeisDefT}
    T^{(\alpha_0)}(z)=\frac12\normord{a(z)^2}+\frac{\alpha_0}{2}\pd a(z),\quad \alpha_0\in\CC.
  \end{equation}
  The central charge is \(c^{(\alpha_0)}=1-3\alpha_0^2\).
\end{defn}

The \ope{} \eqref{eq:hope} implies that the modes of the Laurent expansion
\begin{equation}
  a(z)=\sum_{n\in\ZZ}a_nz^{-n-1}
\end{equation}
satisfy the commutation relations \eqref{eq:heiscomrel} of the Heisenberg algebra $\halg$.  As \(\halg\)-modules, the \(\hvoa{\alpha_0}\) are isomorphic to \(\Fock{0}\), for all \(\alpha_0\in\CC\).
Note that the choice of energy momentum tensor \eqref{eq:HeisDefT}
turns the Fock spaces \(\Fock{p}\) into Virasoro modules via 
\begin{equation}
  L_n=\frac12 \sum_{m\in\ZZ} \normord{a_m a_{n-m}}-\frac{\alpha_0}{2}(n+1)a_n.
\end{equation}
This action determines the conformal weight of the \hwv{}
\(\ketb{p}\in\Fock{p}\) to be \(h_p=\frac{1}{2}p(p-\alpha_0)\).

\subsubsection{The free fermion algebras \(\falg{\epsilon}\)} \label{sec:FreeFermion}

The \emph{free fermion algebras} are a pair of infinite-dimensional complex Lie superalgebras parametrised, as with the $N=1$ superconformal algebras, by $\epsilon \in \set{0, \frac{1}{2}}$:
\begin{equation}
  \falg{\epsilon}=\:\bigoplus_{\mathclap{n\in\ZZ+\epsilon}}\:\CC b_n\oplus\CC\wun.
\end{equation}
The Lie brackets are
\begin{equation} \label{eq:ffcomrel}
  \{b_m,b_n\}=\delta_{m+n}\wun
\end{equation}
and $\wun$ is again central and will be identified with the identity when acting on $\falg{}$-modules. As with the \(N=1\) superconformal algebra, the \(\falg{\sfrac{1}{2}}\)-modules constitute the \ns{} sector and the \(\falg{0}\)-modules the Ramond sector.

The free fermion algebra $\falg{\sfrac{1}{2}}$ admits the triangular decomposition
\begin{equation}
  \falg{\sfrac{1}{2}}^\pm=\bigoplus_{n>0} \CC b_{\pm n},
  \quad
    \falg{\sfrac{1}{2}}^0=\CC\wun ,
\end{equation}
leading, via \(\falg{\sfrac{1}{2}}^\ge=\falg{\sfrac{1}{2}}^+\oplus\falg{\sfrac{1}{2}}^0\), to the \ns{} Verma module
\begin{equation}
  \NSFock=\Ind{\falg{\sfrac{1}{2}}}{\falg{\sfrac{1}{2}}^\ge}\CC\NSket.
\end{equation}
Here, $\CC\NSket$ is the $1$-dimensional \(\falg{\sfrac{1}{2}}^\ge\)-module characterised by $\NSket$ having even parity and
\begin{equation}
  \wun\NSket=\NSket,\quad \falg{\sfrac{1}{2}}^+ \NSket=0.
\end{equation}
This Verma module, together with its parity-reversed counterpart, are simple and are the only \ns{} Verma modules.  They are called \emph{\ns{} Fock spaces}.

The algebra $\falg{0}$ similarly admits a generalised triangular decomposition
\begin{equation}
\falg{0}^\pm=\bigoplus_{n>0} \CC b_{\pm n}, \quad \falg{0}^0=\CC b_0\oplus\CC\wun
\end{equation}
in which $\acomm{b_0}{b_0} = \wun$.  There is a unique simple $\ZZ_2$-graded $\falg{0}^0$-module $\CC \Rket \oplus \CC b_0 \Rket$ on which $\wun$ acts as the identity.  In particular, this module is isomorphic to its parity-reversed counterpart.  Extending this to a module over $\falg{0}^{\ge} = \falg{0}^{+} \oplus \falg{0}^{0}$, by letting $\falg{0}^{+}$ act as zero, the corresponding generalised Verma module is
\begin{equation}
  \RFock=\Ind{\falg{0}}{\falg{0}^\ge}\brac{\CC\Rket\oplus\CC b_0\Rket}.
\end{equation}
This generalised Verma module is also unique (hence invariant under parity reversal).  It is simple as a $\ZZ_2$-graded $\falg{0}$-module and is called the \emph{Ramond Fock space}.

\begin{defn}
  The \emph{free fermion \vosa{}} \(\fvosa\) is the unique \vosa{} that is strongly generated by an
  odd parity field \(b(z)\), has
  the defining \ope{}
  \begin{equation}
    b(z)b(w)\sim\frac{1}{z-w},
  \end{equation}
  and satisfies
  no additional relations beyond
  those required by the \vosa{} axioms. The
  energy-momentum tensor is
  \begin{equation}
    T^{(\falg{})}(z)=\frac12\normord{\pd b(z) b(z)}
  \end{equation}
  and the central charge is \(c^{(\falg{})}=\frac{1}{2}\).
\end{defn}

The modes of the (generalised) Laurent expansion
\begin{equation}
b(z) = \:\sum_{\mathclap{n \in \ZZ + \epsilon}}\: b_n z^{-n-\sfrac{1}{2}}
\end{equation}
satisfy the commutation relations \eqref{eq:ffcomrel} of
the free fermion algebra \(\falg{\epsilon}\). 
As an \(\falg{\sfrac{1}{2}}\)-module, \(\fvosa\) is isomorphic to \(\NSFock\).

\subsubsection{Realising the universal \(N=1\) \vosas{}} \label{sec:N=1FFR}

Although neither \(\hvoa{\alpha_0}\) nor \(\fvosa\) contain an \(N=1\) \vosa{} individually (for instance, neither has a primary field of conformal weight $\frac{3}{2}$), their tensor product \(\ffvoa{\alpha_0}=\hvoa{\alpha_0}\otimes\fvosa\) does.

\begin{prop}\label{thm:ffr}
  Whenever \(\alpha_0^2=\frac12-\frac{c}{3}\), there exists an embedding of \vosas{} \(\svc{c}\ira \ffvoa{\alpha_0}\) that is uniquely determined by the assignment
  \begin{align} \label{eq:ffr}
      T(z)\longmapsto \frac12\normord{a(z)a(z)}+\frac{\alpha_0}{2}\pd a(z)+\frac12\normord{\pd b(z) b(z)},\quad
      G(z)\longmapsto a(z)b(z)+\alpha_0\pd b(z).
  \end{align}
  We omit the tensor product symbols for brevity, identifying $a$ with $a \otimes \wun$ and $b$ with $\wun \otimes b$.
\end{prop}
\begin{proof}
  The image of $T$   is $T^{(\alpha_0)}+T^{(\falg{})}$, the standard choice of energy momentum tensor for the tensor product of \vosas{}.  By explicit computation, one verifies that the images in \eqref{eq:ffr} satisfy the \opes{} \eqref{eq:n1ope} with \(c=\frac32-3\alpha_0^2\).  The assignment \eqref{eq:ffr} thus induces a \vosa{} homomorphism. This homomorphism is obviously an embedding when $\svc{c}$ is simple.  When \(\svc{c}\) is not simple, hence \(c=c_{p_+,p_-}\) for some integers \(p_+,p_-\) satisfying \(p_->p_+\ge2\), \(p_--p_+\in2\ZZ\) and \(\gcd\{p_+,\frac{p_--p_+}{2}\}=1\) (\cref{prop:N=1Simple}), this follows easily from a result of Iohara and Koga \cite[Theorem~4.1]{IohFoc03}.  In detail, \(\ffvoa{\alpha_0}\) is isomorphic to \(\Fock{0}\otimes\NSFock\) as a \(\halg\otimes\falg{\sfrac{1}{2}}\)-module and \eqref{eq:ffr} endows the latter with the structure of a \(\svir{\sfrac{1}{2}}\)-module.  Iohara and Koga determined this structure, showing in particular that the $\svc{c}$-submodule generated by $\ketb{0} \otimes \NSket$ is, for $c=c_{p_+,p_-}$, a length $2$ \hwm{} whose socle is generated by a \sv{} of conformal weight greater than $\frac{1}{2}$.  This submodule is clearly the image of the homomorphism \eqref{eq:ffr} and standard \hw{} theory proves that it is isomorphic to \(\svc{c_{p_+,p_-}}\).  Thus, \eqref{eq:ffr} induces an embedding and the proof is complete.
\end{proof}

\begin{rmk}
The non-trivial part of this proof is to show that, for minimal model central charges, the image of the \sv{} $\vsv{p_+,p_-} \in \svc{c_{p_+,p_-}}$ is non-zero in the free field realisation.  This non-vanishing is crucial for what follows as we shall construct this image, rather than $\vsv{p_+,p_-}$ itself, and use it to classify the modules of the minimal model.
\end{rmk}

\noindent Now that we have established that \(\svc{c}\) embeds into \(\ffvoa{\alpha_0}\), we will identify the fields of \(\svc{c}\) with their images under \eqref{eq:ffr} in \(\ffvoa{\alpha_0}\).  From here on, we will assume that $\alpha_0$ determines the central charge $c=\frac{3}{2}-3\alpha_0^2$.

The tensor product modules
\begin{equation}
  \bNSFock{p}=\Fock{p}\otimes\NSFock,\quad
  \bRFock{p}=\Fock{p}\otimes\RFock\,; \quad p\in\CC,
\end{equation}
are \(\ffvoa{\alpha_0}\)-modules and so are also \(\svc{c}\)-modules. Their \hwvs{}, denoted by
\begin{equation}
\NSffket{p}=\ketb{p}\otimes\NSket,\quad
\Rffket{p}=\ketb{p}\otimes\Rket,
\end{equation}
have conformal weights
\begin{equation} \label{eq:DefConfWts}
  h_p^{\NS}=\frac{1}{2}p(p-\alpha_0),\quad
  h_p^{\R}=\frac{1}{2}p(p-\alpha_0)+\frac{1}{16},
\end{equation}
respectively.

\subsection{Screening operators} \label{sec:Screenings}

A special feature of the Heisenberg algebra is that it allows one to define so-called \emph{vertex operators}.  These, in turn, allow one to construct screening operators for the (non-super) Virasoro minimal models \cite{DotScreen84,TsuFoc86}.  This construction was generalised to the \ns{} and Ramond algebras in \cite{BerSup85,KatoMatsu88,MusFin88}.  We summarise this generalisation here, following \cite{IohFoc03}.

Extend the Heisenberg algebra \(\halg\) by a generator \(\hat a\) satisfying
the relations
\begin{equation}
  \comm{a_m}{\hat a}=\delta_{m,0}, \quad \comm{\hat a}{\wun} = 0.
\end{equation}
A \emph{vertex operator} is the operator-valued formal power series
\begin{equation}
  \vop{p}{z}= \ee^{p\hat a}z^{p a_0}
  \prod_{m\ge1}\left[\exp\left(p\frac{a_{-m}}{m}z^m\right)
    \exp\left(-p\frac{a_{m}}{m}z^{-m}\right)\right],
\end{equation}
which defines a linear map
\begin{equation}
  \vop{p}{z}\colon\Fock{q}\rightarrow \Fock{p+q}\powser{z,z^{-1}}\,z^{pq},
\end{equation}
after identifying \(\ee^{p\hat a}\ketb{q}\) with \(\ketb{p+q}\).  Note that we have grouped the factors of the product such that the exponentials within the square brackets commute for distinct values of $m$.  For later use, we record that the composition of \(k\) vertex operators is given by
\begin{align}\label{eq:vertcomp}
  \vop{p_1}{z_1}\cdots \vop{p_k}{z_k}&=\ee^{\hat a\sum_{i=1}^k p_i}
  \prod_{\mathclap{1\le i<j\le k}} \:
  (z_i-z_j)^{p_ip_j}\cdot
  \prod_{i=1}^k z_i^{p_i a_0}\cdot
  \prod_{m\ge1}\sqbrac*{\exp\brac[\bigg]{\frac{a_{-m}}{m}\sum_{i=1}^k p_i z_i^m}
    \exp\brac[\bigg]{-\frac{a_{m}}{m}\sum_{i=1}^k p_i z_i^{-m}}}.
\end{align}

A standard computation reveals that the
vertex operators \(\vop{p}{z}\) are primary fields of the free boson \voas{} $\hvoa{\alpha_0}$, of Heisenberg weight \(p\) and conformal weight \(h_p=\frac{1}{2}p(p-\alpha_0)\):
\begin{equation}
  a(z)\vop{p}{w}\sim \frac{p\vop{p}{w}}{z-w},\quad
  T^{(\alpha_0)}(z)\vop{p}{w}\sim \frac{h_p\vop{p}{w}}{(z-w)^2}
  +\frac{\pd\vop{p}{w}}{z-w}.
\end{equation}
The vertex operators of immediate interest here are those with $h_p = \frac{1}{2}$ since they are the building blocks of the screening operators introduced below.
This quadratic equation has solutions
\begin{equation}
  p=\alpha_\pm=\frac{1}{2} \left(\alpha_0\pm\sqrt{\alpha_0^2+4}\right),
\end{equation}
which satisfy $\alpha_+\alpha_-=-1$ and
$\alpha_++\alpha_-=\alpha_0$.

\begin{rmk}
	For the central charges $c_{p_+,p_-}$ of the $N=1$ minimal models, we may take the free field data to be
	\begin{equation} \label{eq:DefAlphas}
		\alpha_+ = \sqrt{\frac{p_-}{p_+}}, \quad \alpha_- = -\sqrt{\frac{p_+}{p_-}}, \quad \alpha_0 = \frac{p_- - p_+}{\sqrt{p_+ p_-}}.
	\end{equation}
	Of course, we may also swap $p_+$ and $p_-$ in these formulae.
\end{rmk}

\begin{defn}
A \emph{screening field} for a free field realisation $\VOA{V} \ira
  \VOA{W}$ is a field of the free field \vosa{} $\VOA{W}$, that is, a
    field corresponding to a vector in a module over \(\VOA{W}\), whose \opes{} with
  the fields of $\VOA{V}$ have singular parts that are total derivatives.  It
  suffices to check this for the generating fields of $\VOA{V}$.
\end{defn}
\begin{prop}
  Both
  \begin{equation} \label{eq:DefScr}
    \scrf{+}{z}=b(z)\vop{\alpha_+}{z}\quad\text{and}\quad
    \scrf{-}{z}=b(z)\vop{\alpha_-}{z}
  \end{equation}
  are screening fields for the free field realisation \eqref{eq:ffr} of \(\svc{c}\) in \(\ffvoa{\alpha_0}\):
  \begin{align}
    T(z)\scrf{\pm}{w}\sim \pd_w\frac{\scrf{\pm}{w}}{z-w},\quad
    G(z)\scrf{\pm}{w}\sim\frac{1}{\alpha_\pm} \pd_w \frac{\vop{\alpha_{\pm}}{w}}{z-w}.
  \end{align}
\end{prop}

The appeal of screening fields for a given free field realisation $\VOA{V} \ira \VOA{W}$ is that their residues, when well defined, commute with the action of $\VOA{V}$.  These residues, referred to as \emph{screening operators}, therefore define $\VOA{V}$-module homomorphisms. In particular, they map singular vectors to singular vectors (or zero) and are thus a convenient tool for explicitly constructing singular vectors of $\VOA{V}$-modules.

Consider the following composition of \(k\) screening fields \(\scrf{\pm}{z}\):
\begin{multline} \label{eq:ScrComp'}
  \scrf{\pm}{z_1}\cdots\scrf{\pm}{z_k}=
  b(z_1)\cdots b(z_k)\:
 \ee^{k \alpha_\pm \hat a}
 \prod_{\mathclap{1\le i< j\le k}}\:(z_i-z_j)^{\alpha_\pm^2}\cdot\prod_{i=1}^k z_i^{\alpha_\pm a_0}\\
 \cdot\prod_{m\ge1}\sqbrac*{\exp\brac[\bigg]{\alpha_\pm\frac{a_{-m}}{m}\sum_{i=1}^k z_i^m}
    \exp\brac[\bigg]{-\alpha_\pm\frac{a_{m}}{m}\sum_{i=1}^k z_i^{-m}}}.
\end{multline}
This differs from the analogous compositions required for non-supersymmetric \voas{} \cite{RidJac14,RidSlJac15} in that permuting fermions is skew-symmetric, rather than symmetric, which is problematic if one intends to apply symmetric function techniques.
This can easily be remedied by factoring out the Vandermonde
determinant \(\van{z}=\prod_{1\le i<j\le k}(z_i-z_j)\): 
\begin{align}
  \prod_{\mathclap{1\le i< j\le k}}\:(z_i-z_j)^{\alpha_\pm^2}=
  \van{z}\prod_{\mathclap{1\le i< j\le k}}\:(z_i-z_j)^{\alpha_\pm^2-1}
  =\van{z}\prod_{\mathclap{1\le i\neq j\le k}}\:(z_i-z_j)^{\sfrac{(\alpha_\pm^2-1)}{2}}
\end{align}
(we have suppressed a complex phase in the second equality).  Noting that \((\alpha_\pm^2-1)/2=-\alpha_0/2\alpha_\mp\), this allows us to rewrite \eqref{eq:ScrComp'} in the form
\begin{multline}\label{eq:ScrComp}
  \scrf{\pm}{z_1}\cdots\scrf{\pm}{z_k}=
  \ee^{k \alpha_\pm \hat a}
  \prod_{\mathclap{1\le i\neq j\le k}}\
  \:\brac[\Big]{1-\frac{z_i}{z_j}}^{-\sfrac{\alpha_0}{2\alpha_\mp}}
  \cdot
  \prod_{i=1}^kz_i^{\alpha_\pm a_0+(k-1)\sfrac{(\alpha_\pm^2-1)}{2}}\\
   \cdot \van{z} b(z_1)\cdots b(z_k)
     \prod_{m\ge1}\sqbrac*{\exp\brac[\bigg]{\alpha_{\pm}\frac{a_{-m}}{m}\sum_{i=1}^k z_i^m}
    \exp\brac[\bigg]{-\alpha_{\pm}\frac{a_{m}}{m}\sum_{i=1}^k z_i^{-m}}},
\end{multline}
where the skew-symmetry of the fermion fields is now countered by that of
$\van{z}$.

To define screening operators as integrals of compositions of screening
fields, there need to exist cycles over which to integrate. The obstruction
to the existence of such cycles lies in the multivaluedness of the second product of the \rhs{} of \eqref{eq:ScrComp}. If the exponent \(\alpha_\pm a_0
+(k-1)(\alpha_\pm^2-1)/2\) evaluates to an integer, when we act on a \ns{} free
field module \(\bNSFock{q}\) (so $a_0$ is replaced by $q$), 
then there exists such a cycle \(\Gamma(k,\alpha_0)\), generically unique in homology (up to normalisation) and
constructed in \cite{TsuFoc86}. These cycles are homologically equivalent to the cycles over which one integrates in the theory of
symmetric polynomials to define inner products --- see \cite[Sec.~3]{TsuExt13} for details.
The actual construction of the cycles \(\cyc{k}{\alpha_0}\) is 
rather subtle and we refer the interested reader to \cite{TsuFoc86} for the complete picture.
                                
We mention that when acting on a Ramond free field module \(\bRFock{q}\), the
  cycles \(\Gamma(k,\alpha_0)\) exist when \(\alpha_\pm q
+(k-1)(\alpha_\pm^2-1)/2\) evaluates to a half integer (to compensate for the
half integer exponents of the free fermion fields).  However, screening
operators between Ramond free field modules shall not concern us in what follows.

\begin{defn}\label{sec:scrdef}
	For \(\mu=\frac{1}{2}(1-k)\alpha_\pm+\frac{1}{2}(1-k-2\ell)\alpha_\mp\),
	where \(k\) and \(\ell\) are integers and \(k>0\) is positive,
	the screening operator \(\scrs{\pm}{k}\colon \bNSFock{\mu} \rightarrow \bNSFock{\mu+k\alpha_\pm}\) is
	well defined as the $\svc{c}$-module   homomorphism defined by
	\begin{equation}
		\scrs{\pm}{k}=\int_{\cyc{k}{\alpha_0}}\scrf{\pm}{z_1}\cdots\scrf{\pm}{z_k}\:
		\dd z_1\cdots \dd z_k,
	\end{equation}
	meaning that the cycle \(\cyc{k}{\alpha_0}\) exists.		We choose to normalise this cycle	such that 
	\begin{equation}\label{eq:intnorm}
		\int_{\cyc{k}{\alpha_0}} \prod_{1\le i\neq j\le k} \brac[\Big]{1-\frac{z_i}{z_j}}^{-\sfrac{\alpha_0}{2\alpha_\mp}} \:
		\frac{\dd z_1\cdots \dd z_k}{z_1\cdots z_k} =1.
	\end{equation}
\end{defn}
\noindent We shall lighten notation in what follows by suppressing the cycle \(\cyc{k}{\alpha_0}\) in all integrals.
\begin{rmk}
  As previously stated, the two factors \(\van{z} b(z_1)\cdots b(z_k)\) and
  \begin{equation}
    \prod_{m\ge1}\sqbrac*{\exp\brac[\bigg]{\alpha_{\pm}\frac{a_{-m}}{m}\sum_{i=1}^k z_i^m}
      \exp\brac[\bigg]{-\alpha_{\pm}\frac{a_{m}}{m}\sum_{i=1}^k z_i^{-m}}}
  \end{equation}
  that appear on the \rhs{} of \eqref{eq:ScrComp} are both invariant under permuting the \(z_i\).  The action of the screening operators \(\scrs{\pm}{k}\) can thus be evaluated using the well studied
  family of inner products of symmetric polynomials defined by
  \begin{equation}\label{eq:intprod}
    \jprod{f,g}{k}{t}= \int
    \prod_{1\le i\neq j\le k}
    \brac[\Big]{1-\frac{z_i}{z_j}}^{1/t}
    f(z_1^{-1},\dots,z_k^{-1})g(z_1,\dots,z_k)\:
    \frac{\dd z_1\cdots \dd z_k}{z_1\cdots z_k},
  \end{equation}
  where \(f\) and \(g\) are symmetric polynomials and \(t\in\CC \setminus \set{0}\).
  The Jack symmetric polynomials \(\fjack{\lambda}{t}{z}\) are
  orthogonal with respect to the inner product labelled by $t$
  --- see \cref{sec:Jack} for more details on the role that Jack
  polynomials will play here.
\end{rmk}

\subsection{Correlation functions} \label{sec:Correlators}

In this section, we review some standard results about correlation functions for free bosons and fermions that will be important in later sections.

\subsubsection{Heisenberg correlation functions} \label{sec:CorrBoson}

Let \(\Fock{p}^*\) be the graded dual of the \hw{} $\halg$-module \(\Fock{p}\). Then, \(\Fock{p}^*\) is a \lw{} right $\halg$-module generated by a \lwv{} $\brab{p}$ satisfying 
\begin{equation}
  \braketb{p}{p} = 1, \quad \brab{p} \halg^- = 0.
\end{equation}
It is convenient to extend the domain of the functionals in \(\Fock{p}^*\) to all Fock spaces \(\Fock{q}\), \(q\in\CC\), but to have them act trivially unless \(q = p\).

\begin{defn}
  Let \(B\) be any combination of normally ordered products of free bosons \(a(z)\), vertex operators \(\vop{p}{z}\) and their derivatives.  The \emph{free boson correlation function} in $\Fock{p}$ is then defined to be $\bracketb{p}{B}{p}$.
\end{defn}
\begin{prop}
  The correlation function of \(k\) vertex operators is given by
  \begin{equation}
    \bracketb{p}{\vop{p_1}{z_1}\cdots \vop{p_k}{z_k}}{p}=\delta_{p_1+\cdots+p_k,0}
    \prod_{\mathclap{1\le i<j\le k}}\:(z_i-z_j)^{p_i p_j}\cdot\prod_{i=1}^kz_i^{pp_i}.
  \end{equation}
\end{prop}
\begin{proof}
  This follows directly from the vertex operator composition formula \eqref{eq:vertcomp},
  \begin{align}
    \bracketb{p}{\vop{p_1}{z_1}\cdots \vop{p_k}{z_k}}{p}
    &=\prod_{\mathclap{1\le i<j\le k}}\:(z_i-z_j)^{p_i p_j}\cdot\prod_{i=1}^kz_i^{pp_i}\cdot\braketb{p}{p+p_1+\cdots+p_k},
  \end{align}
  and noting that $\braketb{p}{p+p_1+\cdots+p_k} = \delta_{p_1+\cdots+p_k,0}$.
\end{proof}

\subsubsection{Free fermion correlation functions} \label{sec:CorrFermion}

Let \(\brac*{\NSFock}^*\) and \(\brac*{\RFock}^*\) be the graded duals of \(\NSFock\) and \(\RFock\), respectively, and let \(\NSbra\) and \(\Rbra\) be the respective dual \lwvs{}:
\begin{equation}
  \braketb{\NS}{\NS} = \braketb{\R}{\R} = 1, \quad
  \NSbra \falg{\sfrac{1}{2}}^-  =0, \quad
  \Rbra \falg{0}^- =0.
\end{equation}
\begin{defn}
  Let \(F\) be any combination of normally ordered products of free fermions \(b(z)\) and their derivatives.  The \emph{free fermion correlation functions} $\NScorrfn{F}$ and $\Rcorrfn{F}$ are then defined to be
  \begin{equation}
    \NScorrfn{F}=\bracketb{\mathrm{NS}}{F}{\mathrm{NS}},\quad
    \Rcorrfn{F}=\bracketb{\mathrm{R}}{F}{\mathrm{R}}.
  \end{equation}
\end{defn}

Often, free fermion correlation functions can be conveniently expressed in terms of pfaffians. The determinant of a skew-symmetric matrix \(A = -A^\intercal\) can always be written as the square of a polynomial in the coefficients of \(A\). This polynomial is, up to an unimportant sign ambiguity, the pfaffian \(\pf(A)\) of \(A\). For later convenience, we give two equivalent definitions.
\begin{defn} \label{def:pf}
  Let \(A\) be a \(2n\times2n\) skew-symmetric matrix, so that \(A\) is uniquely determined by its upper-triangular entries \(A_{i,j}\), $i<j$. 
  We shall write $A = (A_{i,j})_{1 \le i < j \le 2n}$ to indicate a skew-symmetric matrix $A$ with given upper-triangular entries.
  \begin{enumerate}
  \item Define the pfaffian of the \(0\times0\) matrix to be $1$.  The pfaffian of $A$ can then be defined recursively by
    \begin{equation}\label{eq:recpf}
      \pf(A)=\sum_{\substack{j=1\\j\neq i}}^{2n}(-1)^{i+j+\theta(j-i)}
      A_{i,j}\pf(A_{\hat{\imath},\hat{\jmath}}),
    \end{equation}
    where the row index $i$ may be chosen arbitrarily, \(A_{\hat{\imath},\hat{\jmath}}\) denotes the matrix \(A\) with the
    \(i\)-th and \(j\)-th rows and columns removed, and
    \begin{equation}
      \theta(x)=
      \begin{cases*}
        1& if \(x>0\), \\
        0& if \(x<0\)
      \end{cases*}
    \end{equation}
    is the Heaviside step function.  In particular, $i=1$ gives the simplified formula
    \begin{equation} \label{eq:recpf'}
      \pf(A) = \sum_{j=2}^{2n} (-1)^j A_{1,j} \pf(A_{\hat{1}, \hat{\jmath}}).
    \end{equation}
  \item Alternatively, an explicit definition of the pfaffian is
  \begin{equation}\label{eq:explpf}
    \pf(A)=\sum_{\sigma\in\Pi}\sgn(\sigma)\prod_{i=1}^n A_{\sigma(2i-1),\sigma(2i)},
  \end{equation}
  where \(\Pi\) is the set of all permutations of the set \(\{1,\dots,2n\}\) that, in Cauchy notation, can be written in the form
  \begin{equation}
    \sigma=
    \begin{bmatrix}
      1&2&3&4&\dots&2n-1&2n\\
      i_1&j_1&i_2&j_2&\dots&i_n&j_n
    \end{bmatrix},\quad i_1<i_2<\dots<i_n,\ i_k<j_k.
  \end{equation}
  Note that this implies that \(i_1=1\).
  \end{enumerate}
\end{defn}

\begin{prop}
  \leavevmode
  \begin{enumerate}
  \item\label{item:odd} The correlation function of an odd number of free fermions is zero in both the Neveu-Schwarz and Ramond sectors.
  \item\label{item:NSeven} In the Neveu-Schwarz sector, the correlation function of \(2n\) free fermions is
    \begin{align}\label{eq:NSpf}
      \NScorrfn{b(z_1)\cdots b(z_{2n})}=\pf\brac*{\frac{1}{z_i-z_j}}_{1\le i<j\le 2n}.
    \end{align}
  \item\label{item:Reven} In the Ramond sector, the correlation function of \(2n\) free fermions is 
    \begin{align}\label{eq:Rpf}
      \Rcorrfn{b(z_1)\cdots b(z_{2n})}=2^{-n}\prod_{i=1}^{2n}z_i^{-\sfrac12}\cdot
      \pf\brac*{\frac{z_i+z_j}{z_i-z_j}}_{1\le i<j\le 2n}.
    \end{align}
  \end{enumerate}
\end{prop}
\begin{proof}
  A correlation function vanishes if its argument is odd, thus for an odd number of fermions, proving \ref{item:odd}.

  We prove \ref{item:NSeven} inductively using the recursive definition of the pfaffian. For $n=0$, the correlation function reduces to $\braketb{\NS}{\NS} = 1$, in agreement with \eqref{eq:NSpf}.  For $n>0$, we expand $b(z_1)$ and then commute its modes to the right:
  \begin{align}
    \NScorrfn{b(z_1)\cdots b(z_{2n})} &= \:\sum_{\mathclap{m \ge \sfrac{1}{2}}}\: \NSbracket{b_m b(z_2) \cdots b(z_{2n})} z_1^{-m-\sfrac{1}{2}} \notag \\
    &= \sum_{j=2}^{2n} (-1)^j \NSbracket{b(z_2) \cdots \widehat{b(z_j)} \cdots b(z_{2n})} \:\sum_{\mathclap{m \ge \sfrac{1}{2}}}\: z_1^{-m-\sfrac{1}{2}} z_j^{m-\sfrac{1}{2}} \notag \\
    &= \sum_{j=2}^{2n} (-1)^j \frac{1}{z_1 - z_j} \NScorrfn{\widehat{b(z_1)} b(z_2) \cdots \widehat{b(z_j)} \cdots b(z_{2n})}.
  \end{align}
  As in \cref{def:pf}, hats denote omission.  Formula \eqref{eq:NSpf} now follows inductively from \eqref{eq:recpf'}.

  Proving \ref{item:Reven} requires a little more work.  Define
  \begin{subequations} \label{eq:RCorrDefFG}
  \begin{align}
    f_n(z_1, \ldots, z_{2n}) &= 2^n \prod_{i=1}^{2n} z_i^{1/2} \cdot \Rcorrfn{b(z_1) \cdots b(z_{2n})}, \label{eq:RCorrDefF} \\
    g_n(z_2, \ldots, z_{2n}) &= \lim_{z_1 \to \infty} f_n(z_1, \ldots, z_{2n})
    = 2^n \prod_{i=2}^{2n} z_i^{1/2} \cdot \Rcorrfn{b_0 b(z_2) \cdots b(z_{2n})} \label{eq:RCorrDefG}
  \end{align}
  \end{subequations}
  and note that $f_0 = 1$ and $f_1(z_1, z_2) = \frac{z_1+z_2}{z_1-z_2}$.  These form the base cases for the assertion that
  \begin{equation} \label{eq:RCorrAssertion}
    f_n(z_1, \dots, z_{2n}) = \pf \brac*{\frac{z_i+z_j}{z_i-z_j}}_{1 \le i < j \le 2n},
  \end{equation}
  which we shall prove by induction on $n$.

  Assume therefore that $n \ge 2$ and that
  \begin{subequations} \label{eq:RCorrAssumption}
  \begin{equation} \label{eq:RCorrAssumptionF}
    f_{n-1}(z_3, \ldots, z_{2n}) = \pf \brac*{\frac{z_i+z_j}{z_i-z_j}}_{3 \le i < j \le 2n}
    = \sum_{j=4}^{2n} (-1)^j \frac{z_3+z_j}{z_3-z_j} f_{n-2}(z_4, \ldots, \widehat{z_j} \ldots, z_{2n}),
  \end{equation}
  using \eqref{eq:recpf'}, from which it follows that
  \begin{equation} \label{eq:RCorrAssumptionG}
    g_{n-1}(z_4, \ldots, z_{2n}) = \sum_{j=4}^{2n} (-1)^j f_{n-2}(z_4, \ldots, \widehat{z_j} \ldots, z_{2n}),
  \end{equation}
  \end{subequations}
  using \eqref{eq:RCorrDefG}.  If we could show that
  \begin{equation} \label{eq:RCorrToBeProved}
    g_n(z_2, \ldots, z_{2n}) = \sum_{j=2}^{2n} (-1)^j f_{n-1}(\widehat{z_1}, z_2, \ldots, \widehat{z_j} \ldots, z_{2n}),
  \end{equation}
  then we would be able to prove \eqref{eq:RCorrAssertion} by expanding $b(z_1)$ in \eqref{eq:RCorrDefF} as follows:
  \begin{align}
    f_n(z_1, \ldots, z_{2n})
    &= 2^n \prod_{i=2}^{2n} z_i^{1/2} \cdot \Rcorrfn{b_0 b(z_2) \cdots b(z_{2n})} + 2^n \prod_{i=2}^{2n} z_i^{1/2} \cdot \sum_{m \ge 1} \Rcorrfn{b_m b(z_2) \cdots b(z_{2n})} z_1^{-m} \notag \\
    &= g_n(z_2, \ldots, z_{2n}) + 2^n \sum_{j=2}^{2n} (-1)^j \prod_{\substack{i=2 \\ i \neq j}}^{2n} z_i^{1/2} \cdot \Rcorrfn{b(z_2) \cdots \widehat{b(z_j)} \cdots b(z_{2n})} \frac{z_j}{z_1-z_j} \notag \\
    &= \sum_{j=2}^{2n} (-1)^j f_{n-1}(\widehat{z_1}, z_2, \ldots, \widehat{z_j} \ldots, z_{2n}) + \sum_{j=2}^{2n} (-1)^j \frac{2z_j}{z_1-z_j} f_{n-1}(\widehat{z_1}, z_2, \ldots, \widehat{z_j}, \ldots, z_{2n}) \notag \\
    &= \sum_{j=2}^{2n} (-1)^j \frac{z_1+z_j}{z_1-z_j} f_{n-1}(\widehat{z_1}, z_2, \ldots, \widehat{z_j}, \ldots, z_{2n}) = \pf \brac*{\frac{z_i+z_j}{z_i-z_j}}_{1 \le i < j \le 2n}.
  \end{align}
  Here, we have also used \eqref{eq:recpf'}, \eqref{eq:RCorrDefFG} and \eqref{eq:RCorrAssumption}.

  To complete the proof, we therefore need to show \eqref{eq:RCorrToBeProved}.  By expanding $b(z_2)$ in \eqref{eq:RCorrDefG} as before, we arrive at
  \begin{equation} \label{eq:RCorrExpandingG}
    g_n(z_2, \ldots, z_{2n}) = f_{n-1}(z_3, \ldots, z_{2n}) + \sum_{j=3}^{2n} (-1)^{j-1} \frac{2z_j}{z_2-z_j} g_{n-1}(z_3, \ldots, \widehat{z_j}, \ldots, z_{2n}).
  \end{equation}
  Using the assumption \eqref{eq:RCorrAssumptionG}, the second term on the \rhs{} can be brought to the more symmetric form
  \begin{equation}
    \sum_{j=3}^{2n} \sum_{\substack{k=3 \\ k \neq j}}^{2n} (-1)^{j+k+\theta(k-j)} \frac{2z_j}{z_2-z_j} f_{n-2}(z_3, \ldots, \widehat{z_j}, \ldots, \widehat{z_k}, \ldots z_{2n}).
  \end{equation}
  The Heaviside function here arises because $z_j$ is omitted.  It is important because it implies that the above double sum would vanish if we replaced $\frac{2z_j}{z_2-z_j}$ by a constant --- the term with $(j,k) = (r,s)$ cancels that with $(j,k) = (s,r)$, for $r \neq s$.  Because of this, we are free to add $1$ to $\frac{2z_j}{z_2-z_j}$, thereby replacing it by $\frac{z_2+z_j}{z_2-z_j}$.  Comparing with \eqref{eq:RCorrAssumptionF}, in the form
  \begin{equation}
    f_{n-1}(z_2, \ldots, \widehat{z_k}, \ldots, z_{2n}) = \sum_{\substack{j=3 \\ j \neq k}}^{2n} (-1)^{j-1+\theta(j-k)} \frac{z_2+z_j}{z_2-z_j} f_{n-2}(z_3, \ldots, \widehat{z_j}, \ldots, \widehat{z_k}, \ldots z_{2n}),
  \end{equation}
  we can rewrite \eqref{eq:RCorrExpandingG} as
  \begin{align}
    g_n(z_2, \ldots, z_{2n}) &= f_{n-1}(\widehat{z_2}, z_3, \ldots, z_{2n}) + \sum_{k=3}^{2n} (-1)^k f_{n-1}(z_2, \ldots, \widehat{z_k}, \ldots, z_{2n}) \notag \\
    &= \sum_{k=2}^{2n} (-1)^k f_{n-1}(\widehat{z_1}, z_2, \ldots, \widehat{z_k}, \ldots, z_{2n}),
  \end{align}
  which is \eqref{eq:RCorrToBeProved}.
\end{proof}
  
\section{From fermionic correlators to Jack symmetric functions} \label{sec:Jack}

In this section, we relate the free fermion correlators that were just calculated to Jack polynomials. The latter turn out to be
unusual in that their parameter is \(-3\).
Such negative rational parameters are usually not permitted because they lead to singularities in the definition of Jack polynomials (with their standard normalisations \cite{MacSym95}).
In \cite{Feidif02}, Feigin, Jimbo, Miwa and Mukhin showed that Jack polynomials with negative rational parameter are well defined only for special, so-called admissible, partitions. Moreover, these admissible Jack polynomials were found to span a certain ideal of the ring of symmetric polynomials. 
We review this theory and then show how these results enable one to derive useful identities relating free fermion correlators and Jack symmetric polynomials. These identities will be crucial for the classification results of the next section.

The standard reference for symmetric functions and polynomials is Macdonald's seminal book \cite{MacSym95}. There, the theory of Jack symmetric functions and polynomials is deduced, sometimes implicitly, from that of the more general Macdonald functions and polynomials.
A short explicit summary of the properties of Jack symmetric functions that we shall require may also be found in \cite[App.~A]{RidSlJac15}.

\subsection{Admissible partitions} \label{sec:AdmPart}

We begin by discussing an important class of partitions, a special case of that introduced in \cite{Feidif02}, establishing the basic properties that we shall later need.
\begin{defn}
  \leavevmode
  \begin{itemize}
  \item Let \(\pi_{\ell}\) be the set of partitions of all integers whose length is at most \(\ell\).
    A partition \(\lambda\in\pi_\ell\) is \emph{admissible} if
    \begin{equation}
      \lambda_i-\lambda_{i+2}\ge 2,\quad 1\le i\le \ell-2.
    \end{equation}
		Note that a partition whose length is strictly less than $\ell$ is understood to be padded with zeroes so that $\lambda_{\ell} = 0$.
  \item For $\ell \ge 2$ and \(n_1, n_2 \in \ZZ\) such that $0 \le n_2 \le n_1 \le n_2 + 2$, denote by \(\admp{\ell}{n_1,n_2} \in \pi_\ell\) the admissible partition whose parts are
    \begin{equation}
      \admp{\ell}{n_1,n_2}_{\ell-1}=n_1;\quad
      \admp{\ell}{n_1,n_2}_{\ell}=n_2;\quad
      \admp{\ell}{n_1,n_2}_{i}=\admp{\ell}{n_1,n_2}_{i+2}+2,
      \quad 1\le i\le \ell-2\,;
    \end{equation}
    that is, $\admp{\ell}{n_1,n_2} = [\ldots, n_1+4, n_2+4, n_1+2, n_2+2, n_1, n_2]$.
  \end{itemize}
\end{defn}
\noindent Thus, \(\admp{\ell}{n_1,n_2}\) is the unique, minimal weight, length (at most) \(\ell\), admissible partition whose last two parts are \(n_1\) and \(n_2\) (in that order).  Its weight is
\begin{equation} \label{eq:AdmWt}
\abs[\big]{\admp{\ell}{n_1,n_2}} =
\begin{cases*}
\frac{1}{2} \ell n_1 + \frac{1}{2} \ell n_2 + \frac{1}{2} \ell (\ell-2), & if \(\ell\) is even, \\
\frac{1}{2} (\ell-1) n_1 + \frac{1}{2} (\ell+1) n_2 + \frac{1}{2} (\ell-1)^2, & if \(\ell\) is odd.
\end{cases*}
\end{equation}

A partition \(\lambda\) is said to be \emph{bounded from above} by another partition \(\mu\) if $\lambda_i\le\mu_i$, for all $i$ (appending zeroes to the end of $\lambda$ if necessary).  This relation will be denoted by \(\lambda\subseteq \mu\).  In this circumstance, one also says that \(\mu\) is \emph{bounded from below} by \(\lambda\). For example, \(\admp{\ell}{0,0}\) bounds every admissible partition in $\pi_\ell$ from below. The more familiar dominance ordering, wherein $\lambda$ is \emph{dominated} by $\mu$ if $\lambda_1 + \cdots + \lambda_i \le \mu_1 + \cdots + \mu_i$, for all $i$, only applies to partitions of equal weights and will be denoted by $\lambda \le \mu$.

\begin{lem}\label{thm:admissiblebounds}
  If $\ell \ge 2$ and \(\mu\in\pi_\ell\) is dominated by the admissible partition \(\admp{\ell}{m,m}\), for some $m \ge 0$, then
  \begin{enumerate}
  \item\label{thm:admbnd1} for \(\ell=2n\) even, the parts of \(\mu\) satisfy $\mu_{2j-1}\ge m+n-j$ and $\mu_{2j} \ge m+n-j$, for $1\le j \le n$;
  \item\label{thm:admbnd2} for \(\ell=2n-1\) odd, the parts of \(\mu\) satisfy $\mu_{2j-1}\ge m+n-j$, for $1\le j \le n$, and $\mu_{2j} \ge m+n-j-1$, for $1\le j \le n-1$.
  \end{enumerate}
	If $\mu$ is instead dominated by any admissible partition $\lambda$ which is bounded from below by \(\admp{\ell}{m+1,m}\), then
  \begin{enumerate}[resume]
  \item\label{thm:admbnd3} for \(\ell=2n\) even, the parts of \(\mu\) satisfy $\mu_{2j-1}\ge m+n-j+1$ and $\mu_{2j} \ge m+n-j$, for $1\le j \le n$;
  \item\label{thm:admbnd4} for \(\ell=2n-1\) odd, the parts of \(\mu\) satisfy $\mu_{2j-1}\ge m+n-j$, for $1\le j \le n$, and $\mu_{2j} \ge m+n-j$, for $1\le j \le n-1$.
  \end{enumerate}
\end{lem}
\noindent It is useful for the proof to first note that the admissible partitions $\admp{\ell}{0,0}$ and $\admp{\ell}{1,0}$ have the following parts:
\begin{equation} \label{eq:admpParts}
\admp{2n}{0,0}_{2j-1} = \admp{2n}{0,0}_{2j} = \admp{2n-1}{0,0}_{2j-1} = 2(n-j), \quad \admp{2n-1}{0,0}_{2j} = 2(n-j-1); \quad \admp{\ell}{1,0}_i = \ell-i.
\end{equation}
We also recall the convenient notation $[m^{\ell}]$ for a length $\ell$ partition, all of whose parts are $m$.
\begin{proof}
  The four estimates on the parts of $\mu$ all follow from the same argument. So, let \(\delta\) be either \(\admp{\ell}{1,0}\) or \(\admp{\ell}{0,0}\), \(\lambda \supseteq \delta + [m^{\ell}]\) and $\mu \le \lambda$.  In case $\delta = \admp{\ell}{0,0}$, we should take $\lambda = \delta + [m^{\ell}]$, but this does not affect the proof.  Since the parts of $\mu$ are weakly decreasing and $\mu \le \lambda$, we have
  \begin{equation} \label{eq:muineq}
    (\ell-i+1)\mu_i \ge \sum_{k=i}^\ell \mu_k \ge \sum_{k=i}^\ell \lambda_k \ge \sum_{k=i}^\ell (m + \delta_k) = (\ell-i+1)m + \sum_{k=i}^\ell \delta_k.
  \end{equation}
  The four estimates now follow by evaluating the sum on the \rhs{}.

  We show the first explicitly. Writing $i=2j-1$ or $i=2j$, the sum evaluates to
 \begin{align}
    \sum_{\mathclap{k=2j-1}}^{2n}\:\admp{2n}{0,0}_k = 2(n-j+1)(n-j) \quad \text{or} \quad 
    \sum_{k=2j}^{2n}\admp{2n}{0,0}_k = (2n-2j+1)(n-j) - (n-j),
  \end{align}
  respectively.  The inequality \eqref{eq:muineq} therefore yields
  \begin{equation}
    \begin{aligned}
      2(n-j+1)\mu_{2j-1}&\ge 2(n-j+1)(n-j),\\
      (2n-2j+1)\mu_{2j}&\ge (2n-2j+1)(n-j)-(n-j)
    \end{aligned}
    \quad\Rightarrow\quad \mu_{2j-1}, \mu_{2j}\ge n-j,
  \end{equation}
  as required.  The remaining estimates follow similarly.
\end{proof}

\begin{lem}\label{thm:uniqueadm}
  Let \(\ell\) be a positive integer. Then, for every integer \(k\) satisfying
  \begin{equation}
    \abs*{\admp{\ell}{0,0}} \le k\le \abs*{\admp{\ell}{1,0}},
  \end{equation}
  there exists exactly one admissible weight $k$ partition
  \(\lambda\in\parl{\ell}\) satisfying \(\admp{\ell}{0,0}\subseteq \lambda\subseteq \admp{\ell}{1,0}\). Furthermore, the only
  admissible partition in $\pi_{\ell}$ dominated by \(\lambda\) is \(\lambda\) itself.
\end{lem}
\begin{proof}
  Since \(\lambda\) is bounded above and below by
  \(\admp{\ell}{1,0}\) and \(\admp{\ell}{0,0}\), respectively, its parts must
  satisfy
  \begin{equation}
    \begin{aligned}
      2n-2j+1\ge{} &\lambda_{2j-1}\ge 2n-2j, & 2n-2j \ge{} &\lambda_{2j} \ge 2n-2j, & &\text{if \(\ell = 2n\);} \\
      2n-2j\ge{} &\lambda_{2j-1}\ge 2n-2j, & 2n-2j-1 \ge{} &\lambda_{2j} \ge 2n-2j-2, & &\text{if \(\ell = 2n-1\).}
    \end{aligned}
  \end{equation}
  Every part whose index has the same parity as $\ell$ is thus fixed, while every other part is constrained to take one of two possible values.  Moreover, if $\lambda_i$ takes the larger of its possible values, for some $i$, then $\lambda_{i-2}$ must also take the larger of its possible values, because admissibility requires that $\lambda_{i-2} - \lambda_i \ge 2$.  It follows that for those parts for which there is a choice, there exists an integer \(m\), with \(0 \le m \le \lfloor \ell/2\rfloor\), such that the first \(m\) parts take the larger value, while the remaining \(\lfloor \ell/2 \rfloor - m\) parts take the smaller value.  Clearly, this gives exactly $\lfloor \ell/2 \rfloor +1$ possibilities for $\lambda$, one for every weight \(\abs{\lambda} = k\) between \(\abs*{\admp{\ell}{0,0}}\) and \(\abs*{\admp{\ell}{1,0}}\).

  Now, fix one such $\lambda$ and assume that there exists an admissible partition \(\mu \in \pi_{\ell}\) that is strictly dominated by \(\lambda\), thus $\abs{\mu} = \abs{\lambda}$ and $\mu \le \lambda$, but $\mu \neq \lambda$.  Let \(i\) be the minimal integer such that \(\mu_i<\lambda_i\).  Since \(\mu\) and \(\lambda\) are both bounded below by \(\admp{\ell}{0,0}\), we must have \(\lambda_{i}=\admp{\ell}{1,0}_i\) and \(\mu_{i}=\admp{\ell}{0,0}_i\). As $\mu$ is admissible, it follows that $\mu_j = \admp{\ell}{0,0}_j \le \lambda_j$, for all $j>i$.  This, however, implies that \(\abs{\mu} < \abs{\lambda}\), a contradiction.
\end{proof}
\begin{defn} \label{def:uniqp}
Given $0 \le m \le \lfloor \ell/2 \rfloor$, denote by $\uniqp{\ell}{m}$ the unique admissible partition of length (at most) $\ell$ and weight $k = \abs*{\admp{\ell}{0,0}} + m$ that is bounded below by $\admp{\ell}{0,0}$ and bounded above by $\admp{\ell}{1,0}$, as in \cref{thm:uniqueadm}.
\end{defn}
\noindent We remark that the parts of these partitions are given by
\begin{equation}
  \uniqp{\ell}{m}_i =
  \begin{cases*}
    \admp{\ell}{1,0}_i, & if \(i \le 2m\), \\
    \admp{\ell}{0,0}_i, & otherwise
  \end{cases*}
\end{equation}
and that the parts of $\admp{\ell}{0,0}$ and $\admp{\ell}{1,0}$ were given in \eqref{eq:admpParts}.  To illustrate this, suppose that $\ell=6$ so that
\begin{equation}
\delta^{(6)}(0,0)=[4,4,2,2,0,0] \quad \text{and} \quad \delta^{(6)}(1,0)=[5,4,3,2,1,0].
\end{equation}
It is easy to check that the admissible partitions $\lambda \in \pi_6$ satisfying $\delta^{(6)}(0,0) \subseteq \lambda \subseteq \delta^{(6)}(1,0)$ are
\begin{equation}
	\begin{aligned}
		\abs*{\lambda}&=12:&\lambda&=[4,4,2,2,0,0]=\uniqp{6}{0}, \\
		\abs*{\lambda}&=13:&\lambda&=[5,4,2,2,0,0]=\uniqp{6}{1},
	\end{aligned}
	\qquad
	\begin{aligned}
		\abs*{\lambda}&=14:&\lambda&=[5,4,3,2,0,0]=\uniqp{6}{2}, \\
		\abs*{\lambda}&=15:&\lambda&=[5,4,3,2,1,0]=\uniqp{6}{3},
	\end{aligned}
\end{equation}
in accordance with \cref{thm:uniqueadm}.

To formulate the next result in a concise way, we introduce the following compact notation.
Consider a partition $\mu$ with at most $\ell$ parts 
and let $n$ be an integer such that $n\ge \mu_1$. Then, we set
\begin{equation}
  [n-\mu] =[n-\mu_\ell,n-\mu_{\ell-1},\dots, n-\mu_1].
\end{equation}
Since the parts of \(\mu\) are subtracted in reverse order, the parts of
\([n-\mu]\) are weakly decreasing. Thus, \([n-\mu]\) is a partition and its
weight is \(n\ell - \abs{\mu}\).

\begin{lem}\label{thm:adminversion}
  If \(n\ge\mu_1\), then $\mu$ is an admissible partition if and only if \([n-\mu]\) is.
\end{lem}
\begin{proof}
  Since the map \(\mu\mapsto [n-\mu]\) is an involution
  on the set of partitions of length at most \(\ell\) with first part bounded by
  \(n\), it is sufficient to merely check the ``if'' part.
  Let \(\nu=[n-\mu]\) and assume that \(\mu\) is admissible. Then,
  \begin{equation}
    \nu_i-\nu_{i+2}=(n-\mu_{\ell-i+1})-(n-\mu_{\ell-i-1})=\mu_{\ell-i-1}-\mu_{\ell-i+1}\ge
    2,\quad 1\le i\le \ell-2,
  \end{equation}
  and, thus, \(\nu\) is admissible. 
\end{proof}

\subsection{Jack symmetric polynomials at $t=-3$} \label{sec:FJMM}

We now turn to the relationship between admissible partitions and Jack symmetric polynomials $\jack{\lambda}{t}$, recalling that the latter are parametrised by partitions $\lambda$ and a complex parameter $t$.  It is common to exclude the case where $t$ is rational and negative because the definition \cite{MacSym95} of $\jack{\lambda}{t}$, as a linear combination of monomial symmetric functions $\monsym{\mu}$ with $\mu \le \lambda$, may then fail for some $\lambda$.  More precisely, the coefficients of this linear combination, in the normalisation where $\jack{\lambda}{t}=\monsym{\lambda}+\cdots$, may diverge for $t \in \QQ_{<0}$.

For instance, consider the Jack polynomial labelled by the partition $[2,2]$:
\begin{equation}
\jack{[2,2]}{t}=\monsym{[2,2]}+\frac{2}{(t+1)}\monsym{[2,1,1]}+\frac{12}{(t+1)(t+2)}\monsym{[1,1,1,1]}.
\end{equation}
It clearly diverges at $t=-1$ and $-2$.
Note that if we set $t=-3$ and restrict to three variables (so that $\monsym{[1,1,1,1]}=0$), then it reduces to
\begin{equation}
	\jack{[2,2]}{-3}(z_1,z_2,z_3) = \monsym{[2,2]}(z_1,z_2,z_3)-\monsym{[2,1,1]}(z_1,z_2,z_3) = z_1^2 z_2^2 + z_1^2 z_3^2 + z_2^2 z_3^2 - z_1^2 z_2 z_3 - z_1 z_2^2 z_3 - z_1 z_2 z_3^2,
\end{equation}
which vanishes when $z_1=z_2=z_3$.  Moreover, the partition $[2,2]=[2,2,0]$ is admissible for three variables.
The essential insight of \cite{Feidif02} was to show that Jack polynomials with negative rational parameter $t$ remain well defined if the partitions are restricted by a suitable admissibility condition and that this restriction is interesting because the well defined Jack polynomials always span a space of symmetric polynomials that vanish when a certain number of variables coincide.

\begin{defn}
  Let \(\fsym{\ell}=\mathbb{C}[z_1,\dots,z_\ell]^{\symgp{\ell}}\) be the ring of complex symmetric polynomials in \(\ell\) variables and let \(I_\ell\) denote the ideal of symmetric polynomials \(f\in\fsym{\ell}\) that satisfy
  \begin{equation}
    f(z_1,\dots,z_\ell)=0 \quad \text{whenever} \quad z_1=z_2=z_3.
  \end{equation}
\end{defn}
\begin{thm}[Feigin, Jimbo, Miwa, Mukhin \cite{Feidif02}]
  \leavevmode
  \begin{enumerate}
  \item Let \(\lambda\in\pi_\ell\) be admissible. Then, the Jack polynomial
    \(\fjack{\lambda}{t}{z_1,\dots,z_\ell}\) is well defined at \(t=-3\), that
    is, its expansion into monomial symmetric polynomials has no poles at \(t=-3\).
  \item The Jack polynomials \(\fjack{\lambda}{-3}{z_1,\dots,z_\ell}\), with \(\lambda\) admissible,
    form a basis of the ideal \(I_\ell\).
  \item The ideal \(I_\ell\) is closed under the action of the differential operators
    \begin{equation}\label{eq:diffVirgens}
      \overline{L}_n=-\sum_{i=1}^\ell z_i^{n+1}\frac{\pd}{\pd z_i},\quad  n\ge -1,\ n\in\ZZ\,
    \end{equation}
    and therefore defines a module over a maximal subalgebra of the centreless Virasoro (or Witt) algebra.
  \end{enumerate}
\end{thm}
\noindent We mention that the Jack parameter $\beta$ used in \cite{Feidif02} is related to $t$ by $\beta = t^{-1}$.

\begin{rmk}
The general form of the admissibility condition introduced in \cite{Feidif02} depends upon two integers $k$ and $r$ for which $k+1$ and $r-1$ are coprime.  There, a partition $\lambda \in \pi_{\ell}$ is said to be $(k,r,\ell)$-admissible if
\begin{equation}
\lambda_i-\lambda_{k+i}\ge r,\quad 1\le i\le \ell-k.
\end{equation}
The authors then proved that the Jack polynomial $\jack{\lambda}{t}$, with $t=-(k+1)/(r-1)$, is well defined whenever $\lambda$ satisfies this admissibility condition.  We will only have need of the special case where $k=2$ and $r=2$, corresponding to the Jack parameter being $t=-3$, hence we refer to the $(2,2,\ell)$-admissible partitions as simply being admissible.
\end{rmk}

Recall from \cref{sec:Screenings} that $\van{z_1,\dots,z_\ell}=\prod_{1\le i<j\le \ell}(z_i-z_j)$ denotes the Vandermonde determinant.  As in that section, we will use this determinant here to trade skew-symmetric functions for symmetric ones.
\begin{prop} \label{prop:VanPf=Jack}
  For every positive integer $n$, we have
  \begin{subequations}\label{eq:vandpfaff}
    \begin{align}
      \label{eq:vandpfaffA}
      \van{z_1,\dots,z_{2n}}\pf\brac*{\frac{1}{z_i-z_j}}_{1\le i<j\le
        2n}&=
      \fjack{\admp{2n}{0,0}}{-3}{z_1,\dots,z_{2n}},\\
      \label{eq:vandpfaffB}
      \van{z_1,\dots,z_{2n}}\pf\brac*{\frac{z_i+z_j}{z_i-z_j}}_{1\le
        i<j\le 2n}&= \fjack{\admp{2n}{1,0}}{-3}{z_1,\dots,z_{2n}}.
    \end{align}
  \end{subequations}
\end{prop}
\noindent The identity \eqref{eq:vandpfaffA} appears to have been originally stated in \cite{BerMod07}, though without proof.  To the best of our knowledge, the identity \eqref{eq:vandpfaffB} is new.  We shall prove both identities here for completeness.
\begin{proof}
  We first show that the two \lhss{} of \eqref{eq:vandpfaff} lie in the ideal \(I_{2n}\). Formula \eqref{eq:explpf} implies that these \lhss{} are equal to
  \begin{equation}\label{eq:explpf2}
    \van{z_1,\dots,z_{2n}} \sum_{\sigma\in\Pi} \sgn(\sigma)\prod_{i=1}^n 
    \frac{(z_{\sigma(2i-1)}+z_{\sigma(2i)})^a}{z_{\sigma(2i-1)}-z_{\sigma(2i)}},
  \end{equation}
  where we set $a=0$ for \eqref{eq:vandpfaffA} and $a=1$ for \eqref{eq:vandpfaffB}.
  This product is clearly symmetric and, as each factor \(z_i-z_j\), for \(1\le i< j\le 2n\), appears exactly once in the Vandermonde
  determinant and at most once in the denominator of each summand of the pfaffians, the product is a polynomial.  However, each variable \(z_i\), for \(1\le i\le 2n\), appears exactly once in the denominator of each summand of the pfaffian, where it is paired with a unique \(z_j\) such that the factor $z_i - z_j$ of the summand cancels that of $\van{z_1,\dots,z_{2n}}$.  The product \eqref{eq:explpf2} might therefore be non-zero if $z_i=z_j$.  However, if three variables are equal, then the product must vanish and so the \lhss{} of \eqref{eq:vandpfaff} lie in the ideal \(I_{2n}\).

  We next prove \eqref{eq:vandpfaffA}.  Note first that the total degree of its \lhs{} is \(2n(n-1) = \abs*{\admp{2n}{0,0}}\).  As $\admp{2n}{0,0}$ is the unique admissible partition of this weight, and as this weight is minimal for all admissible partitions in $\pi_{2n}$, this forces the \lhs{} to be proportional to \(\fjack{\admp{2n}{0,0}}{-3}{z_1,\dots,z_{2n}}\). Equality then follows by expanding the \lhs{} in monomial symmetric polynomials $\monsym{\lambda}$ and showing that the coefficient of
  \begin{equation}
  \monsym{\admp{2n}{0,0}} = z_1^{2(n-1)}z_2^{2(n-1)}z_3^{2(n-2)}z_4^{2(n-2)}\cdots + \cdots
  \end{equation}
  is $1$.  The only summand of \eqref{eq:explpf2} which gives rise to this monomial is the one for which \(\sigma=\id\) and its coefficient is indeed $1$.

  For \eqref{eq:vandpfaffB}, the total degree of the \lhs{} is, instead, \(n(2n-1) = \abs*{\admp{2n}{1,0}}\).  Moreover, if we expand \eqref{eq:explpf2} in the $\monsym{\lambda}$, then the maximal exponent appearing in each monomial is bounded above by $2n-1$ and the next-to-maximal exponent by $2n-2$.  The admissible partitions of length $2n$ and weight $n(2n-1)$ all have parts that violate these bounds, except $\admp{2n}{1,0}$.  It therefore follows that the \lhs{} of \eqref{eq:vandpfaffB} is proportional to \(\fjack{\admp{2n}{1,0}}{-3}{z_1,\dots,z_{2n}}\). As the coefficient of the monomial symmetric polynomial \(\monsym{\admp{2n}{1,0}} = z_1^{2n-1}z_2^{2n-2}z_3^{2n-3}\cdots + \cdots\) is
  \begin{equation}
    \sum_{\sigma\in\Pi} \sgn(\sigma)=\pf\left(1\right)_{1\le i<j\le 2n} =1,
  \end{equation}
  the last equality following easily from the recursion formula \eqref{eq:recpf}, the proportionality factor is again $1$.
\end{proof}

\begin{prop}\label{thm:shiftedcor}
  Let \(n\) be a positive integer.  In the Neveu-Schwarz sector, the fermion correlation functions satisfy the following identities:
  \begin{subequations}\label{eq:NSshiftedcor}
    \begin{align}
      \prod_{i=1}^{2n}z_i^{-(2n-2)}\cdot\van{z_1,\dots,z_{2n}}\NScorrfn{b(z_1+w)\cdots
        b(z_{2n}+w)}&=
      \fjack{\admp{2n}{0,0}}{-3}{z_1^{-1},\dots,z_{2n}^{-1}}, \label{eq:NSshcorr1} \\
      \prod_{i=1}^{2n-1}z_i^{-(2n-3)}\cdot\van{z_1,\dots,z_{2n-1}}
      \NScorrfn{b(z_1+w)\cdots b(z_{2n-1}+w)b(w)}&=
      \fjack{\admp{2n-1}{0,0}}{-3}{z_1^{-1},\dots,z_{2n-1}^{-1}}. \label{eq:NSshcorr2}
    \end{align}
  \end{subequations}
  In the Ramond sector, they instead satisfy
  \begin{subequations}\label{eq:Rshiftedcor}
    \begin{align}
      &2^n\prod_{i=1}^{2n}(z_i+w)^{\sfrac12}\cdot\prod_{i=1}^{2n}z_i^{-(2n-1)}\cdot
      \van{z_1,\dots,z_{2n}}\Rcorrfn{b(z_1+w)\cdots b(z_{2n}+w)} \notag \\
      &\mspace{300mu}=\sum_{m=0}^n c_m^{(2n)} \fjack{[2n-1-\uniqp{2n}{m}]}{-3}{z_1^{-1},\dots,z_{2n}^{-1}} w^{n-m}, \label{eq:Rshcorr1} \\
      &2^n\prod_{i=1}^{2n-1}(z_i+w)^{\sfrac12}\cdot\prod_{i=1}^{2n-1}z_i^{-(2n-2)}\cdot
      \van{z_1,\dots,z_{2n-1}}\Rcorrfn{b_0b(z_1+w)\cdots b(z_{2n-1}+w)} \notag \\
      &\mspace{300mu}=\sum_{m=0}^{n-1} c_m^{(2n-1)} \fjack{[2n-2-\uniqp{2n-1}{m}]}{-3}{z_1^{-1},\dots,z_{2n-1}^{-1}} w^{n-1-m}, \label{eq:Rshcorr2}
    \end{align}
  \end{subequations}
  where the $c_\lambda^{(\ell)} \in \CC$ are constants, the $\uniqp{\ell}{m}$ are the partitions of \cref{def:uniqp}, and $[k-\lambda]$ is the involution of \cref{thm:adminversion}.
\end{prop}
\begin{proof}
  For the Neveu-Schwarz sector, we first note that the correlator in \eqref{eq:NSshcorr1} is a translation-invariant pfaffian, by \eqref{eq:NSpf}, hence its \lhs{} does not depend on \(w\). By \cref{prop:VanPf=Jack}, this \lhs{} is
  \begin{equation} \label{eq:NSCorrJack}
    \prod_{i=1}^{2n}z_i^{-(2n-2)}\cdot \fjack{\admp{2n}{0,0}}{-3}{z_1,\dots,z_{2n}}.
  \end{equation}
  Since the first (and therefore maximal) part of any partition dominated by
  \(\admp{2n}{0,0}\) is, at most, \(\admp{2n}{0,0}_1 = 2n-2\), this is the maximal exponent of any \(z_i\) in the expansion of \(\fjack{\admp{2n}{0,0}}{-3}{z_1,\dots,z_{2n}}\) into monomials.  Therefore, \eqref{eq:NSCorrJack} is a symmetric polynomial of total degree \(2n(n-1)\) in the inverted variables \(z_i^{-1}\) that
  lies in the ideal \(I_l\). It is thus proportional to \(\fjack{\admp{2n}{0,0}}{-3}{z_1^{-1},\dots,z_{2n}^{-1}}\). The equality \eqref{eq:NSshcorr1} now follows from the identity
  \begin{equation}
    \prod_{i=1}^{2n}z_i^{-(2n-2)} \cdot \fmonsym{\admp{2n}{0,0}}{z_1,\dots,z_{2n}} = \fmonsym{[2(n-1)-\admp{2n}{0,0}]}{z_1^{-1},\dots,z_{2n}^{-1}}
  \end{equation}
  and noting that \([2(n-1)-\admp{2n}{0,0}]=\admp{2n}{0,0}\).
  
  For \eqref{eq:NSshcorr2}, note first that $\van{z_1, \ldots, z_{2n-1}, 0} = \prod_{i=1}^{2n-1} z_i \cdot \van{z_1, \ldots, z_{2n-1}}$ implies that its \lhs{} is equal to that of \eqref{eq:NSshcorr1}, multiplied by $z_{2n}^{2n-2}$ and then evaluated at $z_{2n} = 0$.  The \lhs{} of \eqref{eq:NSshcorr2} therefore simplifies to
  \begin{equation}
    \prod_{i=1}^{2n-1}z_i^{-(2n-2)}\cdot\fjack{\admp{2n}{0,0}}{-3}{z_1,\dots,z_{2n-1},0} 
    = \prod_{i=1}^{2n-1}z_i^{-(2n-2)}\cdot\fjack{\admp{2n-1}{2,0}}{-3}{z_1,\dots,z_{2n-1}}.
  \end{equation}
  As $[2n-2-\admp{2n-1}{2,0}] = \admp{2n-1}{0,0}$, the equality with the \rhs{} follows by the same reasoning as before. 

  In the Ramond sector, correlation functions are not translation-invariant. Nevertheless, the translation-invariance of the Vandermonde determinant still allows us to conclude that
  \begin{equation} \label{eq:RCorrJack}
    2^n \prod_{i=1}^{2n}\:(z_i+w)^{\sfrac12}\cdot\van{z_1,\dots,z_{2n}}\Rcorrfn{b(z_1+w)\cdots b(z_{2n}+w)}
        =\fjack{\admp{2n}{1,0}}{-3}{z_1+w,\dots,z_{2n}+w},
  \end{equation}
  using \eqref{eq:Rpf} and \eqref{eq:vandpfaffB}.  Expand \(\fjack{\admp{2n}{1,0}}{-3}{z_1+w,\dots,z_{2n}+w}\) in powers of \(w\):
  \begin{equation}
    \fjack{\admp{2n}{1,0}}{-3}{z_1+w,\dots,z_{2n}+w}=\:\sum_{k=0}^{\mathclap{\abs*{\admp{2n}{1,0}}}}\:
    f_k(z_1,\dots,z_{2n})w^k,\quad
    f_k\in \Lambda_{2n}.
  \end{equation}
  Further, since \(\fjack{\admp{2n}{1,0}}{-3}{z_1+w,\dots,z_{2n}+w}\) vanishes if three or more of the \(z_i\) coincide, the same is true for the \(f_k\).  They therefore lie in the ideal \(I_{2n}\) and are thus linear combinations of Jack polynomials
  \(\fjack{\lambda}{-3}{z_1, \ldots, z_{2n}}\), where the $\lambda$ are admissible.  Moreover, the exponents of the \(z_i\) in each polynomial $f_k$ are bounded by the maximal exponents in \(\fjack{\admp{2n}{1,0}}{-3}{z_1,\dots,z_{2n}}\). Thus, the admissible partitions \(\lambda\) that appear when expanding the \(f_k\) into Jack polynomials must be bounded above by \(\admp{2n}{1,0}\). 
  \cref{thm:uniqueadm} now implies that each of these partitions must be one of the $\uniqp{2n}{m}$, with $0 \le m \le n$, of \cref{def:uniqp}.  By comparing total degrees, we arrive at
  \begin{equation} \label{eq:Jack=Jack}
    \fjack{\admp{2n}{1,0}}{-3}{z_1+w,\dots,z_{2n}+w}=
    \sum_{m=0}^n c_m^{(2n)}\fjack{\uniqp{2n}{m}}{-3}{z_1,\dots,z_{2n}} w^{n-m},\quad
    c_m^{(2n)}\in \CC.
  \end{equation}
  Now, \cref{thm:uniqueadm} also states that any admissible partition \(\mu \in \pi_{2n}\) dominated
  by one of the $\uniqp{2n}{m}$ is equal to $\uniqp{2n}{m}$.  The reasoning followed in the \ns{} sector thus applies, yielding
  \begin{equation}
    \prod_{i=1}^{2n}z_i^{-(2n-1)}\cdot\fjack{\admp{2n}{1,0}}{-3}{z_1+w,\dots,z_{2n}+w}=
    \sum_{m=0}^n c_m^{(2n)}\fjack{[2n-1-\uniqp{2n}{m}]}{-3}{z_1^{-1},\dots,z_{2n}^{-1}} w^{n-m}.
  \end{equation}
  Substituting into \eqref{eq:RCorrJack} finally gives \eqref{eq:Rshcorr1}.  Note that setting $w=0$ in \eqref{eq:Jack=Jack} gives
  \begin{equation}
    \fjack{\admp{2n}{1,0}}{-3}{z_1,\dots,z_{2n}} = c_n^{(2n)} \fjack{\uniqp{2n}{n}}{-3}{z_1,\dots,z_{2n}}
    = c_n^{(2n)} \fjack{\admp{2n}{1,0}}{-3}{z_1,\dots,z_{2n}},
  \end{equation}
  hence $c_n^{(2n)} = 1$.

  The Ramond identity \eqref{eq:Rshcorr2} for an odd number of variables can be derived as a limit of
  the case of an even number of variables. Consider \eqref{eq:Rshcorr1} with $w=0$ and $z_i \to z_{i-1}$ (noting that $[2n-1-\uniqp{2n}{n}] = [2n-1-\admp{2n}{1,0}] = \admp{2n}{1,0}$):
  \begin{equation}
    2^n\prod_{i=0}^{2n-1}z_i^{-(2n-1)}\cdot\prod_{i=0}^{2n-1}z_i^{\sfrac12}\cdot\van{z_0,\dots,z_{2n-1}}
    \Rcorrfn{b(z_0)\cdots b(z_{2n-1})} = \fjack{\admp{2n}{1,0}}{-3}{z_0^{-1},\dots,z_{2n-1}^{-1}}.
  \end{equation}
  Taking the limit as \(z_0\rightarrow \infty\) now gives
  \begin{align}
    2^n\prod_{i=1}^{2n-1}z_i^{-(2n-1)}\cdot\prod_{i=1}^{2n-1} z_i^{\sfrac12}\cdot\van{z_1,\ldots z_{2n-1}}\Rcorrfn{b_0b(z_1)\cdots b(z_{2n-1})}
    &=\fjack{\admp{2n}{1,0}}{-3}{0,z_1^{-1},\dots,z_{2n-1}^{-1}} \notag \\
    &=\fjack{\admp{2n-1}{2,1}}{-3}{z_2^{-1},\dots,z_{2n}^{-1}}\\
    \Rightarrow \qquad 2^n\prod_{i=1}^{2n-1} z_i^{\sfrac12} \cdot \van{z_1,\ldots z_{2n-1}} \Rcorrfn{b_0b(z_1)\cdots b(z_{2n-1})}
    &=\fjack{\admp{2n-1}{1,0}}{-3}{z_1,\dots,z_{2n-1}}.
  \end{align}
  This formula, together with the translation-invariance of the Vandermonde determinant, gives the starting point for analysing the case where the \(z_i\) are shifted by \(w\), analogous to \eqref{eq:RCorrJack}.  The rest of the argument is identical to that described above and we shall omit it, remarking only that it shows that $c_{n-1}^{(2n-1)} = 1$.
\end{proof}

\begin{rmk}
Note that combining \eqref{eq:vandpfaffB} with \eqref{eq:Jack=Jack} results in
\begin{align}
\sum_{m=0}^n c_m^{(2n)}\fjack{\uniqp{2n}{m}}{-3}{z_1,\dots,z_{2n}} w^{-m}
&= \van{z_1, \ldots, z_{2n}} w^{-n} \pf\brac*{\frac{z_i+z_j+2w}{z_i-z_j}}_{1\le i<j\le 2n} \notag \\
&=\van{z_1, \ldots, z_{2n}} \pf\brac*{\frac{1}{w}\frac{z_i+z_j}{z_i-z_j} + \frac{2}{z_i-z_j}}_{1\le i<j\le 2n},
\end{align}
so taking $w \to \infty$ yields
\begin{equation}
c_0^{(2n)} \fjack{\admp{2n}{0,0}}{-3}{z_1,\dots,z_{2n}} = \van{z_1, \ldots, z_{2n}} \pf\brac*{\frac{2}{z_i-z_j}}_{1\le i<j\le 2n} = 2^n \fjack{\admp{2n}{0,0}}{-3}{z_1,\dots,z_{2n}},
\end{equation}
by \eqref{eq:vandpfaffA}.  We conclude that $c_0^{(2n)} = 2^n$.  The coefficient \(c_0^{(2n)}\) being non-zero, in turn implies that
the remaining \(c_m^{(2n)}\) are also non-zero. One can see this by showing that the Taylor expansion \eqref{eq:Jack=Jack} is equivalent to
\begin{align}
  \fjack{\admp{2n}{1,0}}{-3}{z_1+w,\dots,z_{2n}+w}=\sum_{m = 0}^n
  \frac{(-\overline{L}_{-1})^m}{m!}
  \fjack{\admp{2n}{1,0}}{-3}{z_1,\dots,z_{2n}} w^m,
\end{align}
where \(\overline{L}_{-1}\) is one of the differential operators of
\eqref{eq:diffVirgens}. Since \(c_0^{(2n)}\neq0\), it follows that
\(\overline{L}_{-1}^n \fjack{\admp{2n}{1,0}}{-3}{z_1,\dots,z_{2n}}\neq0\)
and thus \(\overline{L}_{-1}^m
\fjack{\admp{2n}{1,0}}{-3}{z_1,\dots,z_{(2n)}}\neq0\), for all $0\le m \le n$. Consequently, \(c_{n-m}^{(2n)}\neq0\), for all \(0 \le m\le n\). However, we do not need this result for the classifications that follow.
\end{rmk}

\section{The minimal model spectrum}\label{sec:zhu}

Zhu's algebra \cite{ZhuMod96} formalises the notion of the algebra of zero modes acting on ground states, these being vectors that are annihilated by all modes of strictly positive index \cite{FeiAnn92,FreVer92,MatZer05}. It is the most important tool for classifying modules over \voas{}. The generalisation of Zhu's algebra to modules over a \vosa{} and its twist for the Ramond sector were first formulated in \cite{Kacn1z94,DonTwi98}. We give a brief overview of twisted Zhu algebras, fixing our notation and emphasising the motivation behind the definitions, in \cref{sec:twistedZhu}. Here, we combine this twisted Zhu theory with the symmetric polynomial technology developed in the previous section to classify the simple $\MinMod{p_+}{p_-}$-modules.

\subsection{Constructing the singular vector $\vsv{p_+,p_-}$} \label{sec:SVConstruction}

The main obstacles to being able to determine Zhu's algebra for the $N=1$ minimal model \vosa{}
\(\MinMod{p_+}{p_-}\) are finding an explicit formula for the singular vector \(\vsv{p_+,p_-}\) in the universal $N=1$ \vosa{} \(\svc{c_{p_+,p_-}}\) (see \cref{prop:N=1Simple}) and then evaluating the action of its zero mode, or that of its descendants, on ground states. This would allow one to determine the images of the ideal generated by $\vsv{p_+,p_-}$ in the Zhu algebra \(\zhu{\svc{c_{p_+,p_-}}}\) and the twisted Zhu algebra \(\tzhu{\svc{c_{p_+,p_-}}}\). The untwisted and twisted
Zhu algebras of \(\MinMod{p_+}{p_-}\) are then the quotients of those of \(\svc{c_{p_+,p_-}}\) by the respective images \cite{Kacn1z94,DonTwi98}.

The free field realisation \eqref{eq:ffr} solves both the problem of finding
the singular vector and the problem of evaluating its zero
mode. Recall that the screening operators \(\scrs{\pm}{k}\), constructed in
\cref{sec:scrdef}, are module homomorphisms of the \(N=1\) superconformal
algebra. Acting with either on a \hwv{}
\(\NSffket{q}\) of appropriate highest weight \(q\) thus gives a singular vector or zero.
For definiteness, we will only use the screening operator $\scrs{+}{k}$ in
what follows.  We will also, without loss of generality, assume from here on
that $p_- > p_+$ and that $c_{p_+,p_-}$ is an $N=1$ minimal model central
charge, so that $p_+ \ge 2$, $p_- - p_+ \in 2 \ZZ$ and $\gcd
\set{\frac{1}{2} (p_- - p_+), p_-} = 1$.
\begin{lem}\label{thm:svlem}
  The vector $\scrs{+}{p_+-1}\NSffket{-(p_+-1)\alpha_+}$ is non-zero, hence it may be identified with the \sv{} of \cref{prop:N=1Simple} that generates the maximal ideal of \(\svc{c_{p_+,p_-}}\subset \ffvoa{\alpha_0}\):
  \begin{equation}\label{eq:VSV}
    \vsv{p_+,p_-}=\scrs{+}{p_+-1}\NSffket{-(p_+-1)\alpha_+}.
  \end{equation}
\end{lem}
\noindent The proof of this lemma uses the theory of symmetric polynomials and their infinite-variable limits, the symmetric functions.  For easy visual distinction, we shall denote the infinite alphabet of
variables for symmetric functions by \(y=(y_1,y_2,\dots)\) and the finite alphabet of
variables for symmetric polynomials by \(z=(z_1,\dots,z_n)\).
We will also need the infinite- and finite-variable inner products \(\cjprod{\cdot,\cdot}{t}\) and \(\jprod{\cdot,\cdot}{n}{t}\), referring to \cite[App.~A]{RidSlJac15} for our conventions, see also \eqref{eq:intprod}.
For use in the reasoning below, we recall the identity (see \cite[Eq.~(A.16)]{RidSlJac15} for instance)
\begin{equation}\label{eq:cauchykernel}
    \prod_{m\ge 1}\exp\brac[\Bigg]{\frac{1}{t} \frac{\fpowsum{m}{z}\fpowsum{m}{y}}{m}} = \sum_{\lambda} \fjack{\lambda}{t}{z} \fdjack{\lambda}{t}{y},
\end{equation}
where $\powsum{m}$ is the $m$-th power sum and the \(\fdjack{\lambda}{t}{y}\) are the symmetric functions dual (with respect to \(\cjprod{\cdot,\cdot}{t}\)) to the Jack symmetric functions \(\fjack{\lambda}{t}{y}\).  We shall refer to the \(\fdjack{\lambda}{t}{y}\) as the dual Jack symmetric functions in what follows.

A simple, but very useful, observation concerning 
the ring of symmetric functions \(\Lambda\) 
is that it is isomorphic, as a commutative algebra, to the \uea{} of either
the positive or negative subalgebra, $\halg^+$ or $\halg^-$, of the Heisenberg
algebra. We denote the
corresponding isomorphisms by
\begin{equation}
  \begin{aligned}
    \symiso{+}{\gamma}\colon\Lambda&\longrightarrow \CC[a_1,a_2,\dots],\\
    \powsum{m}&\longmapsto \gamma a_{m},
  \end{aligned}
\qquad
  \begin{aligned}
    \symiso{-}{\gamma}\colon\Lambda&\longrightarrow \CC[a_{-1},a_{-2},\dots],\\
    \powsum{m}&\longmapsto \gamma a_{-m},
  \end{aligned}
  \qquad \gamma\in\CC\setminus\set{0}.
\end{equation}
Our main use for these isomorphisms will be to identify inner products involving Heisenberg generators with the symmetric function
inner product \(\cjprod{\cdot,\cdot}{t}\). For example, one easily verifies in the power sum basis, hence for arbitrary
\(f,g\in \Lambda\), that
\begin{equation}\label{eq:isoprod}
  \cjprod{f,g}{t}=\bracketb{q}{\symiso{+}{t/\gamma}(f)\symiso{-}{\gamma}(g)}{q},
\end{equation}
where the \rhs{} is evaluated in the Fock space \(\Fock{q}\), for any \(q\in\CC\) and any $\gamma \in \CC \setminus \set{0}$.

\begin{proof}[Proof of \cref{thm:svlem}]
  Let \(\ketb{\phi}=\scrs{+}{p_+-1}\NSffket{-(p_+-1)\alpha_+^2}\) denote the \rhs{} of \eqref{eq:VSV}.
  The singular vectors of the free field $\svir{1/2}$-modules 
  \(\bNSFock{q}\)
  of a  given conformal weight are uniquely determined, up to rescaling, if they exist \cite{IohFoc03}.
  Further, \(\scrs{+}{p_+-1}\) is a module homomorphism so
  \(\ketb{\phi}\) is singular if it is non-zero. Therefore, we only need
  to verify that \(\ketb{\phi}\) is non-zero in order to be able to identify it with
  the singular vector \(\vsv{p_+,p_-}\). We do this by 
  evaluating certain matrix elements and checking explicitly that they are non-zero.

  Using formula \eqref{eq:ScrComp} for the composition of screening operators,
  the \rhs{} of \eqref{eq:VSV} simplifies to
  \begin{multline}
    \ketb{\phi}=
    \int \prod_{1\le i\neq j\le p_+-1} \brac[\Big]{1-\frac{z_i}{z_j}}^{-\sfrac{\alpha_0}{2\alpha_-}} \cdot
    \prod_{i=1}^{p_+-1} z_i^{2-\frac{1}{2} (p_++p_-)} \cdot
    \Delta(z) b(z_1)\cdots b(z_{p_+-1}) \\
    \cdot \prod_{m\ge1}\exp\brac[\Bigg]{\alpha_+ \frac{\fpowsum{m}{z}a_{-m}}{m}}\NSffket{0} \:
    \frac{\dd z_1\cdots \dd z_{p_+-1}}{z_1\cdots z_{p_+-1}}.
  \end{multline}
  Using the isomorphism \(\symiso{-}{\sfrac{2}{\alpha_0}}\)
  and the identity
  \eqref{eq:cauchykernel}, this formula can be re-expressed as
  \begin{multline}
    \ketb{\phi}=
    \int \prod_{1\le i\neq j\le p_+-1} \brac[\Big]{1-\frac{z_i}{z_j}}^{-\sfrac{\alpha_0}{2\alpha_-}} \cdot
    \prod_{i=1}^{p_+-1} z_i^{2-\frac{1}{2} (p_++p_-)} \cdot
    \Delta(z) b(z_1)\cdots b(z_{p_+-1}) \\
    \cdot \sum_{\lambda} \fjack{\lambda}{-\sfrac{2\alpha_-}{\alpha_0}}{z}
    \symiso{-}{\sfrac{2}{\alpha_0}} \brac[\Big]{\fdjack{\lambda}{-\sfrac{2\alpha_-}{\alpha_0}}{y}}
    \NSffket{0} \: \frac{\dd z_1\cdots \dd z_{p_+-1}}{z_1\cdots z_{p_+-1}}.
  \end{multline}
  To further evaluate this formula, we distinguish between \(p_+\) being odd or even.

  Suppose first that \(p_+\) is odd.  Then, setting $w=0$ in \eqref{eq:NSshcorr1} gives 
  \begin{equation}\label{eq:fermionmatel1}
    \prod_{i=1}^{p_+-1}z_i^{2-\frac{1}{2}(p_++p_-)} \cdot \Delta(z)\NScorrfn{b(z_1) \cdots b(z_{p_+-1})} = \fjack{\kappa}{-3}{z_1^{-1},\dots,z_{p_+-1}^{-1}},
  \end{equation}
  where
  \(\kappa\) is the admissible partition 
  \begin{equation}
    \kappa= \admp{p_+-1}{0,0} + \sqbrac*{\brac*{\tfrac{1}{2} (p_- - p_+) + 1}^{p_+-1}} = \admp{p_+-1}{\tfrac{1}{2} (p_- - p_+) + 1, \tfrac{1}{2} (p_- - p_+) + 1}.
  \end{equation}
  The non-vanishing of $\ketb{\phi}$ 
  then follows by evaluating the following matrix element as an integral of a product of \ns{} and Heisenberg vacuum correlators:
  \begin{align}
    &\NSffbra{0} \symiso{+}{\sfrac{1}{\alpha_+}}
    \brac[\Big]{\fjack{\kappa}{-\sfrac{2\alpha_-}{\alpha_0}}{y}}
    \ketb{\phi} \nonumber\\
    &\quad= \int \prod_{1\le i\neq j\le p_+-1} \brac[\Big]{1-\frac{z_i}{z_j}}^{-\sfrac{\alpha_0}{2\alpha_-}} \cdot
    \prod_{i=1}^{p_+-1} z_i^{2-\frac{1}{2} (p_++p_-)} \cdot
    \Delta(z) \NScorrfn{b(z_1)\cdots b(z_{p_+-1})} \notag \\
    &\qquad \cdot \sum_{\lambda} \fjack{\lambda}{-\sfrac{2\alpha_-}{\alpha_0}}{z}
    \bracket[\big]{0}
    {\symiso{+}{\sfrac{1}{\alpha_+}} \brac[\Big]{\fjack{\kappa}{-\sfrac{2\alpha_-}{\alpha_0}}{y}}
    \symiso{-}{\sfrac{2}{\alpha_0}} \brac[\Big]{\fdjack{\lambda}{-\sfrac{2\alpha_-}{\alpha_0}}{y}}}
    {0} \: \frac{\dd z_1\cdots \dd z_{p_+-1}}{z_1\cdots z_{p_+-1}} \nonumber \\
    &\quad= \sum_{\lambda} \int \prod_{1\le i\neq j\le p_+-1} \brac[\Big]{1-\frac{z_i}{z_j}}^{-\sfrac{\alpha_0}{2\alpha_-}}
    \cdot \fjack{\kappa}{-3}{z_1^{-1},\dots,z_{p_+-1}^{-1}} \fjack{\lambda}{-\sfrac{2\alpha_-}{\alpha_0}}{z}
    \: \frac{\dd z_1\cdots \dd z_{p_+-1}}{z_1\cdots z_{p_+-1}}
    \cjprod{\jack{\kappa}{-\sfrac{2\alpha_-}{\alpha_0}}, \djack{\lambda}{-\sfrac{2\alpha_-}{\alpha_0}}}{\sfrac{-2\alpha_}{\alpha_0}} \nonumber \\
    &\quad= \jprod{\jack{\kappa}{-3}, \jack{\kappa}{-\sfrac{2\alpha_-}{\alpha_0}}}{p_+-1}{-\sfrac{2\alpha_-}{\alpha_0}}.
  \end{align}
Here, we have used 
the identities \eqref{eq:isoprod} to evaluate the Heisenberg correlator and 
\eqref{eq:fermionmatel1} to evaluate the \ns{} correlator.  We have also used the pairing of the Jack symmetric functions with their duals. 
To further evaluate the matrix element, we remark that
Jack symmetric polynomials have an upper-triangular decomposition, not only into monomial symmetric polynomials, but also into Jack polynomials \(\jack{\mu}{t}\) with a different parameter:
\begin{equation} \label{eq:JackTriDec}
  \fjack{\kappa}{-3}{z}=
  \sum_{\mu\le\kappa} c_{\kappa\mu} \fjack{\mu}{-\sfrac{2\alpha_-}{\alpha_0}}{z},\quad c_{\kappa\mu}\in\CC, \ c_{\kappa\kappa}=1.
\end{equation}
By the orthogonality of Jack polynomials with respect to the finite-variable inner product, the matrix element therefore evaluates to
  \begin{equation}
    \NSffbra{0} \symiso{+}{\sfrac{1}{\alpha_+}}
    \brac[\Big]{\fjack{\kappa}{-\sfrac{2\alpha_-}{\alpha_0}}{y}}
    \ketb{\phi}=
    \jprod{\jack{\kappa}{-3}, \jack{\kappa}{-\sfrac{2\alpha_-}{\alpha_0}}}{p_+-1}{-\sfrac{2\alpha_-}{\alpha_0}}
    = \jprod{\jack{\kappa}{-\sfrac{2\alpha_-}{\alpha_0}}, \jack{\kappa}{-\sfrac{2\alpha_-}{\alpha_0}}}{p_+-1}{-\sfrac{2\alpha_-}{\alpha_0}} \neq 0.
  \end{equation}

When \(p_+\) is even, the non-vanishing of $\ketb{\phi}$ instead follows from 
the non-vanishing of the matrix element
\begin{multline} \label{eq:matel2}
  \NSffbra{0}\symiso{+}{\sfrac{1}{\alpha_+}}
    \brac[\Big]{\fjack{\kappa}{-\sfrac{2\alpha_-}{\alpha_0}}{y}}
    G_{-\sfrac{1}{2}}\ket{\phi}=-(p_+-1)\alpha_+\NSffbra{0}\symiso{+}{\sfrac{1}{\alpha_+}}
    \brac[\Big]{\fjack{\kappa}{-\sfrac{2\alpha_-}{\alpha_0}}{y}}
    \scrs{+}{p_+-1}b_{-\sfrac{1}{2}}\NSffket{-(p_+-1)\alpha_+}.
\end{multline}
Here, we have used that fact that \(G_{-\sfrac{1}{2}}\) commutes with \(\scrs{+}{p_+-1}\) and then evaluated the action of \(G_{-\sfrac{1}{2}}\) on the free field \hwv{} using \eqref{eq:ffr}. 
The $p_+$ even analogue of \eqref{eq:fermionmatel1} is obtained by setting $w=0$ in \eqref{eq:NSshcorr2}: 
\begin{equation}
  \prod_{i=1}^{p_+-1}z_i^{2-\frac{1}{2}(p_--p_+)}\cdot
  \Delta(z)\NScorrfn{b(z_1)\cdots b(z_{p_+-1})b_{-\sfrac{1}{2}}}
  = \fjack{\kappa}{-3}{z_1^{-1},\dots,z_{p_+-1}^{-1}}.
\end{equation}
The admissible partition here is again $\kappa = \admp{p_+-1}{\frac{1}{2}(p_--p_+)+1, \tfrac{1}{2}(p_--p_+)+1}$. 
The matrix element \eqref{eq:matel2} is now seen to be non-zero using exactly the same steps as in the $p_+$ odd case.
\end{proof}

\subsection{The image of $\vsv{p_+,p_-}$ in Zhu's algebras} \label{sec:SVImage}

Recall that the vacuum vector of the universal $N=1$ \vosa{} $\svc{c}$ is denoted by $\Omega_{0,c}$.  Let $\wun$ and $T$ denote the images of $\Omega_{0,c}$ and $L_{-2} \Omega_{0,c}$ in \(\zhu{\svc{c_{p_+,p_-}}}\) and \(\tzhu{\svc{c_{p_+,p_-}}}\).  Let $G$ denote the image of $G_{-3/2} \Omega_{0,c}$ in \(\tzhu{\svc{c_{p_+,p_-}}}\), so that \(G^2=T-\frac{c}{24}\,\wun\).  We trust that using the same symbols to denote elements of both $\zhu{\svc{c}}$ and $\tzhu{\svc{c}}$ (as well as fields in \(\svc{c}\) and $\MinMod{p_+}{p_-}$) will not lead to any confusion.
\begin{prop}[Kac and Wang {\cite[Lemma 3.1]{Kacn1z94}}; Milas {\cite[App.~B]{MiltWeb07}}]\label{thm:universalzhu}
	\leavevmode
  \begin{enumerate}
  \item The Zhu algebra \(\zhu{\svc{c}}\) of \(\svc{c}\) is isomorphic to \(\CC[T]\), for all \(c\in\CC\).
  \item The twisted Zhu algebra \(\tzhu{\svc{c}}\) of \(\svc{c}\) is isomorphic to \(\CC[G]\cong\CC[T]\oplus \CC[T]\,G\), for all \(c\in\CC\).
  \end{enumerate}
\end{prop}

\begin{rmk}
  Recall the finite-dimensional modules \(\N(h,c)\) and
  \(\R(h,c)\) of \eqref{eq:NSverm} and \eqref{eq:Rverm}, respectively, from
  which the Neveu-Schwarz and Ramond Verma modules were induced. By identifying
  the action of \(T\in\zhu{\svc{c}}\) with that of \(L_0\),
  \(\N(h,c)\) becomes a simple module over \(\zhu{\svc{c}}\). Similarly, by
  identifying the actions of \(T\) and \(G\) in \(\tzhu{\svc{c}}\) with those of \(L_0\)
  and \(G_0\), respectively, \(\R(h,c)\) becomes a simple module over
  \(\tzhu{\svc{c}}\). Naturally, the parity reversals of \(\N(h,c)\) and \(\R(h,c)\) are also simple modules over \(\zhu{\svc{c}}\) and \(\tzhu{\svc{c}}\), respectively.
\end{rmk}

Denote by \(\vsv{p_+,p_-}(w)\) the field corresponding to the singular vector \(\vsv{p_+,p_-} \in \svc{c_{p_+,p_-}}\). The action of its zero mode on an arbitrary \hwv{} $\NSffket{q}$ or $\Rffket{q}$ then follows from evaluating the matrix elements
\begin{equation}
  \NSffbra{q} \vsv{p_+,p_-}(w) \NSffket{q},\quad
  \Rffbra{q} \vsv{p_+,p_-}(w) \Rffket{q}.
\end{equation}
A large proportion of this section will be dedicated to evaluating such matrix elements.  First, however, we address a minor subtlety:  if the \sv{} $\vsv{p_+,p_-}$ is odd (fermionic), then the corresponding field has no zero mode when acting on the \ns{} sector.  In this case, it turns out to be sufficient to consider instead the zero mode of the descendant field corresponding to $G_{-\sfrac{1}{2}} \vsv{p_+,p_-}$.
\begin{prop}[Kac and Wang {\cite[Prop.~3.1]{Kacn1z94}}; Milas {\cite[Lemma~9.3]{MiltWeb07}}]
	\leavevmode
  \begin{enumerate}
  \item The image of the maximal proper ideal $\ideal{\vsv{p_+,p_-}} \subset \svc{c_{p_+,p_-}}$ in the Zhu algebra \(\zhu{\svc{c_{p_+,p_-}}}\) is generated by the image of $\vsv{p_+,p_-}$, if $p_+$ is odd, and by the image of $G_{-\sfrac12}\vsv{p_+,p_-}$, if $p_+$ is even.
  \item The image of the maximal proper ideal $\ideal{\vsv{p_+,p_-}} \subset \svc{c_{p_+,p_-}}$ in the twisted Zhu algebra \(\tzhu{\svc{c_{p_+,p_-}}}\) is generated by the image of $\vsv{p_+,p_-}$, regardless of the parity of $p_+$.
  \end{enumerate}
\end{prop}
\noindent Combining this result with that of \cref{thm:universalzhu}, we learn that the problem of identifying the (twisted) Zhu algebra of $\MinMod{p_+}{p_-}$ reduces to the computation of a single polynomial.
\begin{defn}
Let \(F_{p_+,p_-}(T)\) denote the image of $\vsv{p_+,p_-}$, if $p_+$ is odd,
and $G_{-\sfrac{1}{2}} \vsv{p_+,p_-}$, if $p_+$ is even, in the Zhu algebra
\(\zhu{\svc{c_{p_+,p_-}}} \cong \CC[T]\).  
Similarly, let \(F_{p_+,p_-}^{\parity}(G)\) denote the image of
$\vsv{p_+,p_-}$ in the twisted Zhu algebra \(\tzhu{\svc{c_{p_+,p_-}}} \cong \CC[G]\).
\end{defn}
\noindent The (twisted) Zhu algebras of the $N=1$ minimal models are therefore given by the following quotients:
\begin{equation}\label{eq:zhupres}
\zhu{\MinMod{p_+}{p_-}} \cong \frac{\CC[T]}{\ideal{F_{p_+,p_-}(T)}}, \quad
\tzhu{\MinMod{p_+}{p_-}} \cong \frac{\CC[G]}{\ideal{F_{p_+,p_-}^{\parity}(G)}}.
\end{equation}
Determining the polynomials \(F_{p_+,p_-}(T)\) and \(F_{p_+,p_-}^{\parity}(G)\) is thus our main goal.

The conformal weights of $\vsv{p_+,p_-}$ and its $G_{-1/2}$-descendant imply bounds on the lengths of the \PBW{}-ordered monomials that appear when expressing them as descendants of the vacuum $\Omega_{0,c_{p_+,p_-}}$.  These bounds also apply to the number of generating fields $T(z)$ and $G(z)$ that appear in each normally ordered summand of the corresponding fields of $\svc{c_{p_+,p_-}}$ and so they apply to the number of zero modes $L_0$ and $G_0$ (the latter only if acting in the Ramond sector) appearing in each summand of the zero modes of these fields (assumed to be acting on a ground state).  In other words, the bounds on the lengths of the monomials are bounds for the degrees of the images in the (twisted) Zhu algebra, that is, for the degrees of the polynomials \(F_{p_+,p_-}(T)\) and \(F_{p_+,p_-}^{\parity}(G)\).  These bounds are easily determined.
\begin{lem} \label{lem:degrees}
  \leavevmode
  \begin{enumerate}
  \item The degree of \(F_{p_+,p_-}(T)\) is at most $\frac{1}{4} (p_+-1) (p_--1)$, if $p_+$ is odd, and is at most $\frac{1}{4} (p_+-1) (p_--1) + \frac{1}{4}$, if $p_+$ is even.
  \item The degree of \(F_{p_+,p_-}^{\parity}(G)\) is at most $\frac{1}{2} (p_+-1) (p_--1)$, if $p_+$ is odd, and is at most $\frac{1}{2} (p_+-1) (p_--1) - \frac{1}{2}$, if $p_+$ is even.
  \end{enumerate}
\end{lem}
\noindent We remark that these bounds might not be equalities because it is conceivable that the longest \PBW{}-ordered monomial that could possibly appear in $\vsv{p_+,p_-}$ might come with coefficient zero.  Of course, we also have to rule out the possibility that these polynomials vanish identically.
\begin{lem} \label{lem:PolyNonZeroNS}
	$F_{p_+,p_-}(T)$ is not the zero polynomial.
\end{lem}
\begin{proof}
  Suppose that $F_{p_+,p_-} = 0$, so that the Zhu algebras of $\svc{c_{p_+,p_-}}$ and $\MinMod{p_+}{p_-}$ coincide and every simple $\svc{c_{p_+,p_-}}$-module is also an $\MinMod{p_+}{p_-}$-module.  In particular, every simple \ns{} Verma module is then an $\MinMod{p_+}{p_-}$-module.  We shall show that this contradicts the fact that the modes of $\vsv{p_+,p_-}(w)$ annihilate every $\MinMod{p_+}{p_-}$-module.

  Consider the mode of $\vsv{p_+,p_-}(w)$ with index $-\wt \vsv{p_+,p_-}$.  It is non-zero in the \uesa{} of the \ns{} algebra $\svir{1/2}$ because it acts non-trivially on $\svc{c_{p_+,p_-}}$.  Indeed, its projection onto the \uesa{} of the non-positive subalgebra $\svir{1/2}^{\le} = \svir{1/2}^- \oplus \svir{1/2}^0$ is a non-zero linear combination of monomials in the negative modes whose coefficients are polynomials in $L_0$.  Acting with this mode on the \hwv{} of a \ns{} Verma module $\NSVer{h,c_{p_+,p_-}}$ therefore gives the same linear combination of monomials, but where the polynomials in $L_0$ are evaluated at $L_0 = h$.  These evaluations cannot vanish for every simple \ns{} Verma module because these modules correspond to an infinitude of different values for $h$.  Thus, the mode of $\vsv{p_+,p_-}(w)$ with index $-\wt \vsv{p_+,p_-}$ does not annihilate the \hwv{} of some simple \ns{} Verma module.
\end{proof}
\noindent It is possible to generalise this proof to show that $F_{p_+,p_-}^{\parity}(G)$ is likewise non-zero.  However, this requires a technical detour addressing the subtleties of normal ordering in the Ramond sector.  Instead, we prefer to arrive at this non-vanishing as an easy consequence of the calculation of $F_{p_+,p_-}(T)$, see \cref{lem:PolyNonZeroR} below.

\subsection{The untwisted Zhu algebra of $\MinMod{p_+}{p_-}$} \label{sec:ZhuNS}

This subsection is devoted to the derivation of the explicit form of the polynomial \(F_{p_+,p_-}(T)\).  This result will be used to obtain the classification of simple modules in the \ns{} sector in \cref{sec:Spectrum}.  We also include a detailed example in which the general argument used in the proof is contrasted with a brute force computation of the polynomial \(F_{5,7}(T)\).

Before commencing this derivation, it is convenient to fix a few more definitions.
\begin{defn}
  The \emph{Kac table} of the $N=1$ minimal model $\MinMod{p_+}{p_-}$ is the set of pairs
  \begin{equation}
    \kac{p_+,p_-}=\set*{(r,s)\st 1\le r\le p_+-1,\ 1\le s\le p_--1}
  \end{equation}
  and the \emph{Neveu-Schwarz} and \emph{Ramond Kac tables} are the subsets
  \begin{align}
    \NSkac{p_+,p_-}=\set*{(r,s)\in\kac{p_+,p_-}\st r+s\ \text{is even}}\,\quad\text{and}\quad
    \Rkac{p_+,p_-}=\set*{(r,s)\in\kac{p_+,p_-}\st r+s\ \text{is odd}},
  \end{align}
  respectively.
  Let \(\sim\) denote the equivalence relation on $\kac{p_+,p_-}$ given by $(r,s)\sim (r',s')$ if and only if $(r,s)=(r',s')$ or $(r,s)=(p_+-r',p_--s')$.
  Then, the \emph{reduced Kac table}, as well as the \emph{reduced
    Neveu-Schwarz} and \emph{reduced Ramond Kac tables}, are defined to be
  \begin{equation}
    \rkac{p_+,p_-}=\kac{p_+,p_-}/\sim,\quad
    \rNSkac{p_+,p_-}=\NSkac{p_+,p_-}/\sim,\quad\text{and}\quad
    \rRkac{p_+,p_-}=\Rkac{p_+,p_-}/\sim,
  \end{equation}
  respectively.
\end{defn}
For every element \((r,s)\in\kac{p_+,p_-}\), we define the conformal weight
\begin{equation} \label{eq:DefKacWts}
  h_{r,s}=\frac{(rp_--sp_+)^2-(p_--p_+)^2}{8p_+p_-}+\frac{1-(-1)^{r+s}}{32}.
\end{equation}
The second summand on the \rhs{} evaluates to $0$, when $(r,s) \in
\NSkac{p_+,p_-}$, and to $\frac{1}{16}$, when $(r,s) \in \Rkac{p_+,p_-}$.  Note that $h_{r,s} = h_{p_+-r,p_--s}$.  We shall also define
\begin{equation} \label{eq:DefAlphaRS}
	\alpha_{r,s} = \frac{1-r}{2} \alpha_+ + \frac{1-s}{2} \alpha_-,
\end{equation}
where we recall that $\alpha_+$ and $\alpha_-$ were fixed in \eqref{eq:DefAlphas}.  Note that $h_{\alpha_{r,s}}^{\NS}$ and $h_{\alpha_{r,s}}^{\R}$ are both given by \eqref{eq:DefKacWts}, according as to whether $(r,s) \in \kac{p_+,p_-}$ has $r+s$ even or odd, respectively (see \eqref{eq:DefConfWts}).

\begin{thm}\label{thm:NSzhu}
  The polynomial \(F_{p_+,p_-}(T)\) is given, up to an irrelevant proportionality factor, by
  \begin{equation}
    F_{p_+,p_-}(T)= 
    \prod_{(r,s)}(T-h_{r,s}\,\wun),\quad (r,s)\in\rNSkac{p_+,p_-}.
  \end{equation}
\end{thm}
\begin{proof}
  As the polynomial \(F_{p_+,p_-}(T)\) does not vanish identically, it may be determined by evaluation at sufficiently many values, that is, we evaluate the zero mode of
  \(\vsv{p_+,p_-}\) (or \(G_{-\sfrac12}\vsv{p_+,p_-}\)) on candidate
  \hwvs{} of arbitrary conformal weight. These evaluations will be performed using the free field realisation by evaluating the action of the zero mode
  of \(\vsv{p_+,p_-}\) (or \(G_{-\sfrac12}\vsv{p_+,p_-}\)) on a free field
  \hwv{} \(\NSffket{q}\) of conformal weight
  \(h_q^{\NS}=\frac{1}{2}q(q-\alpha_0)\), for arbitrary $q \in \CC$. This \hwv{} will be an
  eigenvector of the zero mode 
  because the vector spans a $1$-dimensional weight space and its eigenvalue 
  is \(F_{p_+,p_-}(h_q^{\NS})\), by definition.
  Up to rescaling, \(F_{p_+,p_-}(h_q^{\NS})\) is uniquely characterised by its roots in \(q\) (with multiplicity).  Finally, as $h_q^{\NS}$ is quadratic in $q$, the roots of \(F_{p_+,p_-}(h_q^{\NS})\) come in pairs:  $q$ is a root if and only if $\alpha_0-q$ is a root.

  Consider first the case when \(p_+\) is odd, so that the singular vector \(\vsv{p_+,p_-}\) has even parity.  Then, \(F_{p_+,p_-}(h_q^{\NS})\) is determined by evaluating the matrix element
  \begin{equation} \label{eq:MatrixElement}
    \NSffbracket{q }{\vsv{p_+,p_-}(w)}{q }=F_{p_+,p_-}(h_q^{\NS}) w^{-(p_+-1)(p_--1)/2}.
  \end{equation}
  Recall that \(\NSffbra{q }\) is the vector dual to \(\NSffket{q }\) and that $\vsv{p_+,p_-}$ is expressed in terms of screening fields in \eqref{eq:VSV}.  The corresponding field is
  \begin{equation}
    \vsv{p_+,p_-}(w) = \int \scrf{+}{z_1+w} \cdots \scrf{+}{z_{p_+-1}+w} \vop{-(p_+-1)\alpha_+}{w} \: \dd z_1 \cdots \dd z_{p_+-1},
  \end{equation}
  hence the matrix element may be expressed as the integral
  \begin{multline}
    \NSffbracket{q }{\vsv{p_+,p_-}(w)}{q } = \int \van{z} \NScorrfn{b(z_1+w)\cdots b(z_{p_+-1}+w)} \\
    \cdot \: \prod_{\mathclap{1\le i\neq j\le p_+-1}} \: (z_i-z_j)^{\brac{\alpha_+^2-1}/2} \cdot \prod_{i=1}^{p_+-1} z_i^{-(p_+-1)\alpha_+^2} (z_i+w)^{\alpha_+ q} \cdot w^{-(p_+-1)\alpha_+ q} \: \dd z_1\cdots \dd z_{p_+-1},
  \end{multline}
  where we have used the definition \eqref{eq:DefScr} of the screening fields and the composition \eqref{eq:vertcomp} of vertex operators, as well as factored out a Vandermonde determinant $\van{z} = \prod_{1 \le i < j \le p_+-1} (z_i-z_j)$.
  Identity \eqref{eq:NSshcorr1}
  of \cref{thm:shiftedcor} can now be used, along with \eqref{eq:DefAlphas}, to rewrite the matrix element as
  \begin{multline}
    \NSffbracket{q }{\vsv{p_+,p_-}(w)}{q}=\int \fjack{\admp{p_+-1}{0,0}}{-3}{z_1^{-1}, \ldots, z_{p_+-1}^{-1}}
        \prod_{i=1}^{p_+-1}z_i^{-1-\brac{p_--p_+} / 2} \\
      \cdot \prod_{1\le i\neq j\le p_+-1} \brac[\Big]{1-\frac{z_i}{z_j}}^{-\alpha_0 / 2\alpha_-}
				\cdot \prod_{i=1}^{p_+-1} \brac*{1+\frac{z_i}{w}}^{\alpha_+ q}
        \: \frac{\dd z_1\cdots \dd z_{p_+-1}}{z_1\cdots z_{p_+-1}},
  \end{multline}
  which we recognise as a (finite-variable) symmetric polynomial inner product, see \eqref{eq:intprod}.  The matrix element thus becomes
  \begin{equation} \label{eq:NSmatrixel}
     \NSffbracket{q }{\vsv{p_+,p_-}(w)}{q } = \jprod{\fjack{\kappa}{-3}{z_1, \ldots, z_{p_+-1}}, \prod_{i=1}^{p_+-1} \brac*{1+\frac{z_i}{w}}^{\alpha_+q }}{p_+-1}{-2\alpha_- / \alpha_0},
  \end{equation}
  where $\kappa = \admp{p_+-1}{\frac{1}{2}(p_--p_+)+1, \tfrac{1}{2}(p_--p_+)+1}$.

  The product in \eqref{eq:NSmatrixel} is most easily expanded into Jack symmetric polynomials using an algebra homomorphism $\Xi_{X}$, which maps symmetric functions to complex numbers, called the specialisation map. This is defined to map each power sum (in an infinite number of variables $y = (y_1, y_2, \ldots)$) to the same $X \in \CC$.  Explicitly, it gives $\Xi_{X}(\fpowsum{k}{y}) = X$, for all $k \ge 1$.  We specialise with $X = -2q / \alpha_0$, as in \cite[Eq.~(A.28)]{RidSlJac15}, and combine this with the homogeneity of Jack symmetric polynomials to obtain
  \begin{align}\label{eq:jackprod}
    \prod_{i=1}^{p_+-1}\brac*{1+\frac{z_i}{w}}^{\alpha_+ q}
    &= \prod_{k \ge 1} \exp \brac[\bigg]{-\alpha_+ q \frac{\fpowsum{k}{-z_1/w, \ldots, -z_{p_+-1}/w}}{k}} \notag \\
    &= \Xi_{-2q / \alpha_0} \sqbrac*{\prod_{k \ge 1} \exp \brac[\bigg]{-\frac{\alpha_0}{2 \alpha_-} \frac{\fpowsum{k}{-z_1/w, \ldots, -z_{p_+-1}/w} \fpowsum{k}{y_1, y_2, \ldots}}{k}}} \notag \\
    &= \sum_\lambda (-w)^{-\abs{\lambda}} \fdjack{\lambda}{-2\alpha_- / \alpha_0}{z_1, \ldots, z_{p_+-1}}
    \Xi_{-2q / \alpha_0} \sqbrac*{\fjack{\lambda}{-2\alpha_- / \alpha_0}{y_1, y_2, \ldots}}.
  \end{align}
  where the sum runs over all partitions.

  On the other hand, the Jack symmetric polynomial in \eqref{eq:NSmatrixel} needs to be expanded into Jack polynomials with parameter $-2 \alpha_- / \alpha_0$, for which we make use of the triangular decomposition \eqref{eq:JackTriDec}.
  Using this and \eqref{eq:jackprod}, the matrix element \eqref{eq:NSmatrixel} now takes the form
  \begin{multline} \label{eq:NSeval}
    \NSffbracket{q}{\vsv{p_+,p_-}(w)}{q}
    = \sum_{\mu\le \kappa} \sum_\lambda c_{\kappa\mu} (-w)^{-\abs{\lambda}}
    \jprod{\jack{\mu}{-2 \alpha_- / \alpha_0}, \djack{\lambda}{-2 \alpha_- / \alpha_0}}{p_+-1}{-2 \alpha_- / \alpha_0}
    \Xi_{-2q / \alpha_0} \sqbrac*{\fjack{\lambda}{-2 \alpha_- / \alpha_0}{y_1, y_2, \ldots}} \\
    =\sum_{\mu\le \kappa} c_{\kappa\mu} \jprod{\jack{\mu}{-2 \alpha_- / \alpha_0}, \djack{\mu}{-2 \alpha_- / \alpha_0}}{p_+-1}{-2 \alpha_- / \alpha_0} \Xi_{-2q / \alpha_0} \sqbrac*{\fjack{\mu}{-2 \alpha_- / \alpha_0}{y_1, y_2, \ldots}} (-w)^{-(p_+-1)(p_--1) / 2},
  \end{multline}
  with the help of \eqref{eq:AdmWt}.  Up to an unimportant sign, we have therefore identified $F_{p_+,p_-}(h_q^{\NS})$ as
  \begin{equation} \label{eq:F=Jack}
	F_{p_+,p_-}(h_q^{\NS}) = \sum_{\mu\le \kappa} c_{\mu}^{\NS} \, \Xi_{-2q / \alpha_0} \sqbrac*{\fjack{\mu}{-2 \alpha_- / \alpha_0}{y_1, y_2, \ldots}},
	\end{equation}
	where the constants
	\begin{equation}
		c_{\mu}^{\NS} = c_{\kappa\mu} \jprod{\jack{\mu}{-2 \alpha_- / \alpha_0}, \djack{\mu}{-2 \alpha_- / \alpha_0}}{p_+-1}{-2 \alpha_- / \alpha_0}
	\end{equation}
	do not depend on $q$.

  The explicit form of the specialised Jack symmetric function in
  \eqref{eq:F=Jack} is (see, for example, \cite[Eq. (A.24)]{RidSlJac15})
  \begin{equation}\label{eq:Jackspec}
		\Xi_X(\jack{\mu}{t}) = \prod_{b\in\mu} \frac{X - l^\prime(b) + t a^\prime(b)}{1 + l(b) + t a(b)},
  \end{equation}
  where $a(b)$, $l(b)$, $a^\prime(b)$ and $l^\prime(b)$ are the arm length, leg length, arm colength and leg colength, respectively, of the box $b$ of the diagram of $\mu$.  This shows that the roots in \(X = -2q / \alpha_0\) of this specialisation only depend on the arm and leg colengths (the denominators in \eqref{eq:Jackspec} are manifestly non-zero as $t = -2 \alpha_- / \alpha_0 = 2p_+ / (p_- - p_+)$ is positive).  We recall that these colengths are the distances $j-1$ (arm) and $i-1$ (leg) from the box $b$ at position $(i,j)$ to the left and top edges of the diagram.  Thus, the boxes that are common to all partitions \(\mu\le\kappa\) will give rise to roots that are common to all summands of \eqref{eq:F=Jack} and are thus roots of $F_{p_+,p_-}(h_q^{\NS})$.  Invoking the estimate \ref{thm:admbnd1} of \cref{thm:admissiblebounds}, we learn that the parts of $\mu$ must satisfy $\mu_{2j-1}, \mu_{2j} \ge \frac{1}{2} (p_-+1)-j$, hence that the boxes common to all the \(\mu\) include those that form the partition
  \begin{equation} \label{eq:DefRho}
    \rho = \sqbrac*{\tfrac{1}{2} (p_- - 1), \tfrac{1}{2} (p_- - 1), \tfrac{1}{2} (p_- - 3), \tfrac{1}{2} (p_- - 3), \ldots, \tfrac{1}{2} (p_- - p_+) + 1, \tfrac{1}{2} (p_- - p_+) + 1}.
  \end{equation}
  Specialising the Jack symmetric polynomial $\jack{\rho}{-2\alpha_- / \alpha_0}$ will thus give some of the roots of the matrix element \eqref{eq:MatrixElement}.
 
  Performing this specialisation (and dropping various irrelevant proportionality factors) yields
  \begin{align}
    \Xi_{-2q / \alpha_0}(\jack{\rho}{-2\alpha_- / \alpha_0})
    &= \prod_{b\in\rho} \brac*{-\frac{2q}{\alpha_0} - l^\prime(b) - \frac{2\alpha_-}{\alpha_0} a^\prime(b)} \nonumber \\
		&= \prod_{i=1}^{\frac{1}{2} (p_+-1)} \prod_{j=1}^{\frac{1}{2} (p_-+1)-i}
		\brac[\big]{q + (i-1) \alpha_0 + (j-1)\alpha_-} \brac[\big]{q + (i-\tfrac{1}{2}) \alpha_0 + (j-1)\alpha_-} \nonumber \\
		&= \prod_{i=1}^{\frac{1}{2} (p_+-1)} \prod_{j=1}^{\frac{1}{2} (p_-+1)-i}
		\brac*{q  - \alpha_{2i-1,2i+2j-3}} \brac*{q  - \alpha_{2i,2i+2j-2}}
		= \underset{r=s \bmod{2}}{\prod_{r=1}^{p_+-1} \prod_{s=r}^{p_--1}} \brac*{q - \alpha_{r,s}}.
  \end{align}
  Thus, \(q =\alpha_{r,s}\) is a root of \(F_{p_+,p_-}(h_q^{\NS})\), for all \((r,s)\in\NSkac{p_+,p_-}\)
  with \(r\le s\). Suppose now that \((r,s)\in\NSkac{p_+,p_-}\) has $r>s$.  Then, \((p_+-r,p_--s)\in\NSkac{p_+,p_-}\) has $p_+-r < p_--s-(p_--p_+) < p_--s$, so $\alpha_{p_+-r,p_--s}$ is a root and, thus, so is \(\alpha_0-\alpha_{p_+-r,p_--s}=\alpha_{r,s}\).  We have therefore established that $q=\alpha_{r,s}$ is a root of \(F_{p_+,p_-}(h_q^{\NS})\), for all $(r,s) \in \NSkac{p_+,p_-}$.  It now follows that \(F_{p_+,p_-}(h_q^{\NS})\) contains the factor
  \begin{equation} \label{eq:q->h}
    \brac*{q -\alpha_{r,s}} \brac*{q - \alpha_{p_+-r,p_--s}} = 2 \brac*{h_q^{\NS} - h_{r,s}},
  \end{equation}
  hence that the conformal weights \(h_{r,s}\), with \((r,s)\in\rNSkac{p_+,p_-}\) are roots of \(F_{p_+,p_-}(T)\). Since the
  degree of \(F_{p_+,p_-}(T)\) is at most \(\frac{1}{4}(p_+-1)(p_--1)\), which coincides with $\abs[\big]{\rNSkac{p_+,p_-}}$,  we have found all the roots and the theorem follows (for \(p_+\) odd).

  The proof for \(p_+\) even uses very similar reasoning.  The main difference is that the matrix element to be evaluated is
  \begin{equation}\label{eq:NSevenmatel}
    \NSffbracket{q}{(G_{-\sfrac12}\vsv{p_+,p_-})(w)}{q}.
  \end{equation}
  As the screening operator \(\scrs{+}{p_+-1}\) is a module homomorphism for the \(N=1\) superconformal algebra, we have
  \begin{align}
		G_{-\sfrac12}\vsv{p_+,p_-} &= G_{-\sfrac12}\scrs{+}{p_+-1}\NSffket{-(p_+-1)\alpha_+}
		= \scrs{+}{p_+-1}G_{-\sfrac12}\NSffket{(p_+-1)\alpha_+} \notag \\
		&= \scrs{+}{p_+-1}(p_+-1)\alpha_+ b_{-\sfrac12}\NSffket{-(p_+-1)\alpha_+}.
  \end{align}
  The matrix element \eqref{eq:NSevenmatel} is therefore proportional to
  \begin{equation}
    \int\NSffbracket{q}{\scrf{+}{z_1+w}\cdots\scrf{+}{z_{p_+-1}+w}b(w)V_{-(p_+-1)\alpha_+}(w)}{q} \: \dd z_1\cdots \dd z_{p_+-1}.
  \end{equation}
  One now follows the same arguments as before, except that \eqref{eq:NSshcorr2} is used to express this matrix element using the Jack polynomial inner product instead of \eqref{eq:NSshcorr1}.  This inner product can then be written as a sum of specialisations of Jack symmetric polynomials, thus allowing one to find the common roots \(q\) of \(F_{p_+,p_-}(h_q^{\NS})\), this time by using estimate \ref{thm:admbnd2} of \cref{thm:admissiblebounds}.  We note that for this case, the partition $\rho$ of \eqref{eq:DefRho} is replaced by
  \begin{equation}
    \sqbrac*{\tfrac{1}{2} p_-, \tfrac{1}{2} p_- - 1, \tfrac{1}{2} p_- - 1, \tfrac{1}{2} p_- - 2, \ldots, \tfrac{1}{2} (p_- - p_+) + 1, \tfrac{1}{2} (p_- - p_+) + 1}.
  \end{equation}
  The rest of the proof is identical.
\end{proof}

We recall that the filtration on Zhu's algebra gave an easy upper bound on the degree of $F_{p_+,p_-}(T)$ (\cref{lem:degrees}).  A consequence of the previous proof is that this bound is, in fact, always saturated.  The saturation of this bound was stated by Kac and Wang in \cite[Thm.~3.1]{Kacn1z94}, without proof.  It seems plausible that a direct proof might be obtained by generalising the method of Astashkevich \cite[Thms.~3.1 and 8.4]{AstStr97} from Verma modules to the \ns{} module underlying the universal \vosa{} $\svc{c_{p_+,p_-}}$.  In any case, this saturation implies the following result.
\begin{cor} \label{cor:saturation}
	If $\vsv{p_+,p_-} \in \svc{c_{p_+,p_-}}$ is expressed as a linear combination of \PBW{}-ordered monomials acting on the \hwv{} $\Omega_{0,c_{p_+,p_-}}$, then the coefficient of $L_{-2}^{(p_+-1)(p_--1)/4}$ is non-zero, if $p_+$ is odd, and that of $L_{-2}^{(p_+-1)(p_--1)/4-3/4} G_{-3/2}$ is non-zero, if $p_+$ is even.
\end{cor}
\noindent We mention that the analogue of this result for universal Virasoro \voas{} is proven in \cite[Lem.~9.6]{IohRep11}.

\medskip

Let us illustrate the method used in the proof of \cref{lem:PolyNonZeroNS} by working out an explicit example, namely the derivation of the allowed free field data $\alpha_{r,s}$ and conformal weights $h_{r,s}$ for the \ns{} sector of the minimal model $\MinMod{5}{7}$.  With $p_+=5$ and $p_-=7$, we have $\kappa = \admp{4}{2,2} = [4,4,2,2]$ and
\begin{equation}
-\frac{2\alpha_-}{\alpha_0}=\frac{2p_+}{p_--p_+}=5.
\end{equation}
The explicit decomposition \eqref{eq:JackTriDec} giving the $c_{\kappa\mu}$ is highly truncated as there are only $p_+ - 1 = 4$ variables:
\begin{align}
\jack{[4,4,2,2]}{-3}(z_1,\ldots, z_4)&=\jack{[4,4,2,2]}{t}(z_1,\ldots, z_4)-\frac{(t+3)}{(t+1)}\jack{[4,3,3,2]}{t}(z_1,\ldots, z_4)+\frac{6(t+4)}{(t+2)}\jack{[3,3,3,3]}{t}(z_1,\ldots, z_4)\nonumber\\
&=\jack{[4,4,2,2]}{5}(z_1,\ldots, z_4)-\frac43\jack{[4,3,3,2]}{5}(z_1,\ldots, z_4)+\frac{54}{7}\jack{[3,3,3,3]}{5}(z_1,\ldots, z_4).
\end{align}
The diagrams of the contributing partitions are
\begin{equation}
  \ydiagram{4,4,2,2}*[*(lightgray)]{3,3,2,2}\;,\qquad
  \ydiagram{4,3,3,2}*[*(lightgray)]{3,3,2,2}\;,\qquad
  \ydiagram{3,3,3,3}*[*(lightgray)]{3,3,2,2}\qquad\Rightarrow\qquad
  \rho=\,\ydiagram[*(lightgray)]{3,3,2,2}\;,
\end{equation}
where $\rho$ is the diagram formed by the boxes that are common to all three diagrams (indicated with shading).  
Each box $b$ of $\rho$ contributes a factor $q-\alpha_{r,s}(b)$ to the polynomial $F_{5,7}(T)$, where the indices $r$ and $s$ for each box are indicated below together with the corresponding conformal weights $h_{r,s}^{\NS}$ and the \ns{} Kac table of conformal weights (with shading indicating the entries determined by $\rho$):
\begin{equation}
  \alpha_{r,s}:\quad
  \ytableausetup{boxsize=2em}
  \begin{ytableau}
    1,1&1,3&1,5\\
    2,2&2,4&2,6\\
    3,3&3,5\\
    4,4&4,6
  \end{ytableau}
  \;,\qquad h_{r,s}^{\NS}:\quad
  \begin{ytableau}
    0&\frac{3}{14}&\frac{8}{7}\\
    \frac{3}{70}&\frac{4}{35}&\frac{9}{10}\\
    \frac{4}{35}&\frac{3}{70}\\
    \frac{3}{14}&0
  \end{ytableau}
  \;,\qquad\NSkac{5,7}:\quad
  \begin{ytableau}
		*(lightgray)0&&*(lightgray)\tfrac{3}{14}&&*(lightgray)\tfrac{8}{7}&\\
		&*(lightgray)\tfrac{3}{70}&&*(lightgray)\tfrac{4}{35}&&*(lightgray)\tfrac{9}{10}\\
		\tfrac{9}{10}&&*(lightgray)\tfrac{4}{35}&&*(lightgray)\tfrac{3}{70}&\\
		&\tfrac{8}{7}&&*(lightgray)\tfrac{3}{14}&&*(lightgray)0
  \end{ytableau}
  \;.
\end{equation}
As noted in the proof of \cref{lem:PolyNonZeroNS}, we obtain all the allowed conformal weights from $\rho$ because of the $\ZZ_2$-symmetry of the Kac table.

It is instructive to see that the two factors $q-\alpha_{r,s}$ that are missed by $\rho$, namely $(r,s) = (3,1)$ and $(4,2)$, do actually appear in $F_{5,7}(h_q^{\NS})$.  This requires the explicit form identified in \eqref{eq:F=Jack}.  For brevity, set
\begin{equation} \label{eq:DefJdJ}
  \JdJ{\mu}{n}{t} = \jprod{\jack{\mu}{t}, \djack{\lambda}{t}}{n}{t}
  = \prod_{b\in\mu} \frac{n-l'(b)+ta'(b)}{n-(l'(b)+1)+t(a'(b)+1)}
\end{equation}
(see \cite[Eq. VI.10.37]{MacSym95}), where we recall that $\jprod{1,1}{n}{t}$ has been normalised to $1$ in \eqref{eq:intnorm}.  With this notation, we have
\begin{multline}
	F_{5,7}(h_q^{\NS}) =\JdJ{[4,4,2,2]}{4}{5}\,\Xi_{-2q / \alpha_0} \sqbrac*{\jack{[4,4,2,2]}{5}} 
	-\frac{4}{3}\JdJ{[4,3,3,2]}{4}{5}\,\Xi_{-2q / \alpha_0} \sqbrac*{\jack{[4,3,3,2]}{5}} \\
	+\frac{54}{7}\JdJ{[3,3,3,3]}{4}{5}\,\Xi_{-2q / \alpha_0} \sqbrac*{\jack{[3,3,3,3]}{5}},
	\end{multline}
where $\alpha_0 = \frac{2}{\sqrt{35}}$.  Let $h_\mu(t)=\prod_{b\in \mu} (1 + l(b) + t a(b))$, so that the specialisation is  
\begin{equation}
	\Xi_X\sqbrac*{\jack{\mu}{t}} = \frac{1}{h_\mu(t)} \prod_{b\in\mu} \brac*{X - l^\prime(b) + t a^\prime(b)}.
\end{equation}
For the three partitions of interest here, we obtain
\begin{equation}
\begin{aligned}
h_{[4,4,2,2]}(t)&=24(t+1)^3(t+2)^3(2t+3)(3t+4)\overset{t=5}=439\,193\,664,\\
h_{[4,3,3,2]}(t)&=8(t+1)^2(t+2)^2(t+3)^2(2t+3)(3t+4)\overset{t=5}=223\,082\,496,\\
h_{[3,3,3,3]}(t)&=96(t+1)^2(t+2)^2(t+3)(t+4)(2t+1)(2t+3)\overset{t=5}=1\,743\,565\,824.
\end{aligned}
\end{equation}
Factoring out the product over the common boxes of $\rho$ now gives
\begin{equation}
F_{5,7}(h_q^{\NS}) = \brac*{-\frac{2}{\alpha_0}}^{12}  \mathcal N_\rho(4;5)   \, F^{(1)} F^{(2)},
\end{equation}
where the roots determined by $\rho$ are bundled into
\begin{equation}
F^{(1)} = \prod_{b\in\rho} \brac*{q + \tfrac{1}{2} \alpha_0 l^\prime(b) + \alpha_- a^\prime(b)}
=\ \prod_{\mathclap{\substack{(r,s)\in \NSkac{5,7} \\ (r,s)\ne(3,1),(4,2)}}}\ (q-\alpha_{r,s})
\end{equation}
and the two other roots are (hidden) in
\begin{multline}
F^{(2)}=  \frac{(q+3\alpha_-) (q+\alpha_0/2+3\alpha_-)}{h_{[4,4,2,2]}(5)} \brac*{\frac{19}{23}\frac{9}{11}}
- \frac{4}{3} \frac{(q+3 \alpha_-)(q+\alpha_0+2\alpha_-)}{h_{[4,3,3,2]}(5)} \brac*{\frac{19}{23}\frac{3}{4}} \\ 
+ \frac{54}{7} \frac{(q+\alpha_0 + 2 \alpha_-)(q+3 \alpha_0/2+2 \alpha_-)}{h_{[3,3,3,3]}(5)} \brac*{\frac{3}{4}\frac{11}{15}}.
\end{multline}
Here, the rational numbers in parentheses are the contributions to \eqref{eq:DefJdJ} from the boxes of each partition that are not in $\rho$.  A brute force simplification now results in
\begin{equation}
F^{(2)} \propto \brac*{q+ \frac{\sqrt{35}}{5}} \brac*{q+\frac{8}{\sqrt{35}}} = (q-\alpha_{3,1}) (q-\alpha_{4,2}),
\end{equation}
so we do indeed recover the two missing roots, albeit at a significant computational cost.

\subsection{The twisted Zhu algebra of $\MinMod{p_+}{p_-}$} \label{sec:ZhuR}

We now turn to the derivation of the polynomial $F_{p_+,p_-}^{\parity}(G)$, required for the classification in the Ramond sector. The first step is to demonstrate that this polynomial does not vanish identically.

\begin{lem} \label{lem:PolyNonZeroR}
	$F_{p_+,p_-}^{\parity}(G)$ is not the zero polynomial.	
\end{lem}
\begin{proof}
	Recall that $F_{p_+,p_-}^{\parity}(G)$ is the expression for the zero mode of $\vsv{p_+,p_-}(w)$ acting on ground states.  For $p_+$ odd, the coefficient of $L_{-2}^{(p_+-1)(p_--1)/4}$ being non-zero in \cref{cor:saturation} implies that the coefficient of $T^{(p_+-1)(p_--1)/4}$, hence that of $G^{(p_+-1)(p_--1)/2}$, in $F_{p_+,p_-}^{\parity}(G)$ is non-zero.  It is easy to check that the other \PBW{} monomials give polynomials in $G$ of (strictly) smaller degree, hence this non-zero term cannot be cancelled and $F_{p_+,p_-}^{\parity}(G)$ is not zero.  A similar argument settles the case for $p_+$ even.
\end{proof}

\begin{thm}\label{thm:Rzhu}
  The polynomial \(F_{p_+,p_-}^\parity(G)\) is given, up to an irrelevant proportionality factor, by
  \begin{equation}
    F_{p_+,p_-}^\parity(G)=
    \begin{cases*}
      \prod_{(r,s)}(T-h_{r,s}\,\wun),\phantom{{}\cdot G}\quad
      (r,s)\in\rRkac{p_+,p_-},& if \(p_+\) is odd,\\
      \prod_{(r,s)}(T-h_{r,s}\,\wun)\cdot G,\quad
      (r,s)\in\rRkac{p_+,p_-}\setminus\set*{(p_+/2,p_-/2)},
      & if \(p_+\) is even,
    \end{cases*}
  \end{equation}
    where \(T=G^2+\frac{1}{24} c_{p_+,p_-} \, \wun\).
\end{thm}
\begin{proof}
  For \(p_+\) odd, the singular vector \(\vsv{p_+,p_-}\) is even and thus, so
  is its image in \(\tzhu{\svc{c_{p_+,p_-}}}\).  It is therefore a polynomial in
  \(T\). The proof then follows the same arguments as in that of \cref{thm:NSzhu}, starting from the matrix element $\Rffbracket{q}{ \vsv{p_+,p_-}(w)}{q}$.  We outline the minor complications that are encountered.  First, the identity \eqref{eq:Rshcorr1} implies that the analogue of \eqref{eq:F=Jack} is the seemingly more complicated expression
  \begin{equation} \label{eq:Ftw=Jack}
    F_{p_+,p_-}^{\parity}(h_q^{\R}) = \sum_{m=0}^{\mathclap{\frac{1}{2}(p_+-1)}} \sum_{\mu\le \kappa(m)} c_{m\mu}^{\R}
    \Xi_{-(2q + \alpha_-) / \alpha_0} \sqbrac*{\fjack{\mu}{-2 \alpha_- / \alpha_0}{y_1, y_2, \ldots}},
  \end{equation}
  where $\kappa(m) = [\frac{1}{2}(p_+ + p_-) - 2 - \uniqp{p_+-1}{m}]$ is an admissible partition, by \cref{thm:adminversion}, the
  \begin{equation}
    c_{m\mu}^{\R}= 
    (-1)^mc_m^{(p_+-1)}c_{\kappa(m) \mu}
    \jprod{\jack{\mu}{-2 \alpha_- / \alpha_0}, \djack{\mu}{-2 \alpha_- / \alpha_0}}{p_+-1}{-2 \alpha_- / \alpha_0}
  \end{equation}
  are constants that do not depend on \(q\), and we have suppressed an unimportant overall power of 2.  We recall from \cref{def:uniqp} that $\uniqp{p_+-1}{m} \subseteq \admp{p_+-1}{1,0}$, for all $m$, hence 
  \begin{equation}
    \kappa(m) \supseteq \sqbrac*{\tfrac{1}{2}(p_+ + p_-) - 2 - \admp{p_+-1}{1,0}} = \admp{p_+-1}{\tfrac{1}{2} (p_- - p_+) + 1, \tfrac{1}{2} (p_- - p_+)}
  \end{equation}
  and so estimate \ref{thm:admbnd3} of \cref{thm:admissiblebounds} applies.  The result is that every $\mu$ in \eqref{eq:Ftw=Jack} is bounded below by
  \begin{equation}
    \rho = \sqbrac*{\tfrac{1}{2} (p_- - 1), \tfrac{1}{2} (p_- - 3), \tfrac{1}{2} (p_- - 3), \tfrac{1}{2} (p_- - 5), \ldots, \tfrac{1}{2} (p_- - p_+ + 2), \tfrac{1}{2} (p_- - p_+)}
  \end{equation}
  and this suffices to conclude the proof as in the \ns{} cases.
  
  For \(p_+\) even, the singular vector is odd, that is, the zero mode of
  \(\vsv{p_+,p_-}(w)\) is parity changing.  Therefore, any matrix element of the form \(\Rffbracket{q}{\vsv{p_+,p_-}(w)}{q}\)
  necessarily vanishes.  To circumvent this, we shall instead evaluate the matrix element \(\Rffbracket{q}{b_0 \vsv{p_+,p_-}(w)}{q}\).  This evaluation proceeds as for $p_+$ odd, using \eqref{eq:Rshcorr2} and estimate \ref{thm:admbnd4} of \cref{thm:admissiblebounds}, the result being (up to irrelevant proportionality constants)
  \begin{align} \label{eq:RPolyEven}
    \Rffbracket{q}{b_0 \vsv{p_+,p_-}(w)}{q}
    &= \underset{r \neq s \bmod{2}}{\prod_{r=1}^{p_+-1} \prod_{s=1}^{p_--1}} \brac*{q - \alpha_{r,s}} \cdot w^{-(p_+-1)(p_--1)/2} \notag \\
    &= \prod_{(r,s)} \brac*{h_q^{\R} - h_{r,s}} \cdot \brac*{q - \alpha_{p_+/2,p_-/2}} w^{-(p_+-1)(p_--1)/2},
  \end{align}
  where the final product is over all $(r,s)$ in $\rRkac{p_+,p_-}$ except $(p_+/2,p_-/2)$.  The reason for this exception is that when $p_+$ is even, the map $(r,s) \mapsto (p_+-r,p_--s)$ has a fixed point in $\Rkac{p_+,p_-}$, hence the factor $q - \alpha_{p_+/2,p_-/2}$ does not pair up, as in \eqref{eq:q->h}, to give a polynomial in $h_q^{\R}$.  The interpretation of this factor in the zero mode algebra (twisted Zhu algebra) is therefore not in terms of $L_0$ ($T$), but in terms of $G_0$ ($G$).  Indeed, the free field realisation \eqref{eq:ffr} gives
  \begin{align}
    \Rffbracket{q}{b_0 G_0}{q}
    = \Rffbracket{q}{b_0 (a_0 b_0 - \tfrac12 \alpha_0 b_0)}{q}
    = \tfrac12 (q-\tfrac{1}{2} \alpha_0)
    =\tfrac12 (q-\alpha_{\sfrac{p_+}{2},\sfrac{p_-}{2}}),
  \end{align}
  making this interpretation explicit and completing the proof.
\end{proof}

\subsection{Classifying modules} \label{sec:Spectrum}

The classification of simple modules over the twisted and untwisted Zhu algebras is now an easy consequence of identifying the polynomials $F_{p_+,p_-}(T)$ and $F_{p_+,p_-}^{\parity}(G)$.

\begin{thm}\label{thm:Zhuclass}
  \leavevmode
  \begin{enumerate}
  \item The Zhu algebra \(\zhu{\MinMod{p_+}{p_-}}\) is semisimple and, up to
    equivalence, its simple $\ZZ_2$-graded modules are exhausted by the \(\N(h_{r,s},c_{p_+,p_-})\), with \((r,s)\in\rNSkac{p_+,p_-}\), and their parity reversals.
  \item The twisted Zhu algebra \(\tzhu{\MinMod{p_+}{p_-}}\) is semisimple and, up to
    equivalence, its simple $\ZZ_2$-graded modules are exhausted by the \(\R(h_{r,s},c_{p_+,p_-})\), with \((r,s)\in\rRkac{p_+,p_-}\),
    and, if \(p_+\) is even, the parity reversal of \(\R(h_{p_+/2,p_-/2},c_{p_+,p_-})\) 
    (the others being isomorphic to their parity-reversed counterparts).
  \end{enumerate}
\end{thm}
\begin{proof}
  The classification of simple modules follows immediately from the
  presentations \eqref{eq:zhupres} and the explicit formulae for the
  polynomials \(F_{p_+,p_-}(T)\) and \(F_{p_+,p_-}^\parity(G)\) in
  \cref{thm:NSzhu,thm:Rzhu}, respectively.  For the
  semisimplicity, we first note that $T$ would have a single eigenvalue when
  acting on any non-split extension of two simple modules because the
  extension would be indecomposable and $T$ is central in both
  \(\zhu{\MinMod{p_+}{p_-}}\) and \(\tzhu{\MinMod{p_+}{p_-}}\).  It follows
  that the two simple modules would need to be isomorphic.  If the simple
  modules were \ns{}, then the self-extension would have to have a
  non-semisimple action of $T$.  Similarly, if the simple modules were Ramond,
  but not isomorphic to \(\R(h_{p_+/2,p_-/2},c_{p_+,p_-})\), then the
  self-extension would have to have a non-semisimple action of $G$, hence a
  non-semisimple action of $T = G^2 + \frac{c}{24} \, \wun$ (because $G$ has
  non-zero eigenvalues).  In both cases, an indecomposable self-extension is
  ruled out because the polynomials \(F_{p_+,p_-}(T)\) and
  \(F_{p_+,p_-}^\parity(G)\) have no repeated factors of the form $T-h$, hence
  $T$ must act semisimply.  The remaining case, where the simple modules are
  Ramond and isomorphic to \(\R(h_{p_+/2,p_-/2},c_{p_+,p_-})\), would have to
  have a non-semisimple action of $G$, but not necessarily of $T$.  However,
  this is ruled out by \(G\) appearing as a single unrepeated factor in \(F_{p_+,p_-}^\parity(G)\).
\end{proof}

Let \(\NSIrr{h,c}\) and \(\RIrr{h,c}\) denote the unique simple quotients of
the \ns{} and Ramond Verma modules \(\NSVer{h,c}\) and \(\RVer{h,c}\), respectively.
 \begin{thm}\label{thm:rationality}
   The \vosa{} \(\MinMod{p_+}{p_-}\) is rational in both the Neveu-Schwarz and Ramond sectors, that is, both sectors have a finite number of simple modules and every $\ZZ_2$-graded module is semisimple.
   \begin{enumerate}
   \item Up to equivalence, the simple $\ZZ_2$-graded $\MinMod{p_+}{p_-}$-modules in the Neveu-Schwarz sector are exhausted by the \(\NSIrr{{h_{r,s},c_{p_+,p_-}}}\), with \((r,s)\in\rNSkac{p_+,p_-}\), 
   and their parity reversals.
   \item Up to equivalence, the simple $\ZZ_2$-graded $\MinMod{p_+}{p_-}$-modules in the Ramond sector are {exhausted} by the \(\RIrr{h_{r,s},c_{p_+,p_-}}\), with \((r,s)\in\rRkac{p_+,p_-}\), and, if $p_+$ is even, the parity reversal of \(\RIrr{h_{p_+/2,p_-/2}}\) (the others being isomorphic to their parity-reversed counterparts).
   \end{enumerate}
 \end{thm}
\begin{proof}
  The classification of simples follows from \cref{thm:Zhuclass} and the usual bijective correspondence between simple modules over the (twisted) Zhu algebra and simple (twisted) modules (with ground states) over the associated \vosa{}.  To show semisimplicity, and thus rationality, one needs to rule out indecomposable extensions of a simple module by another simple. Indecomposable self-extensions are ruled out because they would correspond to indecomposable self-extensions of (twisted) Zhu algebra modules, contradicting the semisimplicity of the latter.  To rule out indecomposable extensions involving two inequivalent simple modules, $\Mod{M}$ and $\Mod{N}$, note that either the indecomposable or its contragredient dual would have to be a \hwm{}.  The conformal weight of the ground states of the submodule, $\Mod{M}$ say, of this \hwm{} would then have to match that of a \sv{} in the Verma module that covers $\Mod{N}$.  However, it is easy to check \cite{IohVer03} that a Verma module with conformal weight in $\NSkac{p_+,p_-}$ or $\Rkac{p_+,p_-}$ never has a descendant \sv{} whose conformal weight is also in $\NSkac{p_+,p_-}$ or $\Rkac{p_+,p_-}$, respectively.  Such extensions therefore do not exist and thus the rationality in both sectors follows.
\end{proof}

\appendix

\section{Twisted Zhu algebras} \label{sec:twistedZhu}

The results of \cite{Kacn1z94} detail Zhu's algebra for untwisted modules over \vosas{}, while \cite{DonTwi98} introduces a version of Zhu's algebra for modules that have been twisted by a (finite-order) automorphism \(\parity\).  For the $N=1$ \vosas{} studied here, the untwisted modules are those in the Neveu-Schwarz sector and the Ramond sector corresponds to the special case in which \(\parity\) is the parity automorphism, defined to act as the identity on the even subspace and minus the identity on the odd subspace.  Throughout this appendix, we shall assume that $\VOA{V}$ is a \vosa{}, graded by conformal weights in $\frac{1}{2} \ZZ$, and that $\parity$ is an order $2$ automorphism of $\VOA{V}$.

Let us say that a vector $v \in \VOA{V}$ is \emph{homogeneous} if it is a simultaneous \(L_0\)- and \(\tau\)-eigenvector and, in this case, define $\wt v$ to be its conformal weight.  Let $\zmsub{V}{0}$ and $\tzmsub{V}{0}$ ($\zmsub{V}{1/2}$ and $\tzmsub{V}{1/2}$) denote the subspaces of $\VOA{V}$ spanned by the homogeneous vectors whose associated fields have integer moding (half-integer moding) when acting on the untwisted and $\parity$-twisted sectors, respectively.  Then, $\zmsub{V}{0}$ and $\zmsub{V}{1/2}$ are the eigenspaces of $(-1)^{2 L_0}$ of eigenvalues $1$ and $-1$, respectively, and we always have $\VOA{V} = \zmsub{V}{0} \oplus \zmsub{V}{1/2} = \tzmsub{V}{0} \oplus \tzmsub{V}{1/2}$.  We give three examples to illustrate this simple, but crucial, definition:
\begin{itemize}
\item Let $\VOA{V}$ be an $N=1$ \vosa{} and let $\parity$ be the parity automorphism.  Then, $\zmsub{V}{0}$ and $\zmsub{V}{1/2}$ are the even and odd subspaces of $\VOA{V}$, respectively, while $\tzmsub{V}{0} = \VOA{V}$ and $\tzmsub{V}{1/2} = 0$.
\item Let $\VOA{V}$ be the \vosa{} associated with symplectic fermions ($\AKMSA{psl}{1}{1}$), or another affine Kac-Moody superalgebra, and let $\tau$ be the parity automorphism.  Then, $\zmsub{V}{0} = \VOA{V}$ and $\zmsub{V}{1/2} = 0$, while $\tzmsub{V}{0}$ and $\tzmsub{V}{1/2}$ are the even and odd subspaces of $\VOA{V}$, respectively.
\item Let $\VOA{V}$ be the bosonic ghost system of central charge $c=-1$, so that the ghost fields have conformal weight $\frac{1}{2}$, and let $\parity = (-1)^{2L_0}$.  Then, $\zmsub{V}{0}$ and $\zmsub{V}{1/2}$ are the subspaces whose non-zero vectors are constructed using an even and odd number of ghost modes, respectively, while $\tzmsub{V}{0} = \VOA{V}$ and $\tzmsub{V}{1/2} = 0$.
\end{itemize}
Of course, the first is the example that is important for this paper.

The untwisted (twisted) Zhu algebra of a \vosa{} \(\VOA{V}\) is, as a vector space, a quotient of the subspace $\zmsub{V}{0}$ ($\tzmsub{V}{0}$).  To characterise these quotients and define the algebra operations, we consider the following bilinear products \cite{ZhuMod96} defined on both $\zmsub{V}{0}$ and $\tzmsub{V}{0}$:
\begin{equation} \label{eq:DefZhuProds}
u \zstar v = \res{z=0} \sqbrac*{u(z) v \frac{(1+z)^{\wt u}}{z}}, \quad u \zcirc v = \res{z=0} \sqbrac*{u(z) v \frac{(1+z)^{\wt u}}{z^2}}.
\end{equation}
Both may be motivated by considering the following generalised commutation relation, see \cite[App.~B]{RidSlJac15}, and assuming that it acts on a ground state:
\begin{equation} \label{eq:GCR}
	\sum_{\ell=0}^{\infty} \binom{\ell+k-1}{k-1} \sqbrac*{u_{n-\ell} v_{-n+\ell} - (-1)^{k+\abs{u}\abs{v}} v_{-n-k-\ell} u_{n+k+\ell}} = \sum_{j=0}^{\infty} \binom{\wt u+n+k-1}{j} (u_{-\wt u+j-k+1}v)_0
\end{equation}
($u$ and $v$ are here assumed to be homogeneous vectors of definite parities $\abs{u}$ and $\abs{v}$, respectively).  Indeed, taking $n=0$ and $k=1$ gives $u_0 v_0 = \sum_j \binom{\wt u}{j} (u_{-\wt u+j}v)_0 = (u\zstar v)_0$ on a ground state.  The product $\zstar$ is therefore just the abstraction of this zero mode action to elements of $\zmsub{V}{0}$ and $\tzmsub{V}{0}$.  Unlike the zero mode action however, this product fails to be associative in general.

Taking instead $n=-1$ and $k=2$ in \eqref{eq:GCR}, we obtain the relations $(u\zcirc v)_0 = \sum_j \binom{\wt u}{j} (u_{-\wt u-1+j}v)_0 = 0$.  Abstracting these identities therefore amounts to the vanishing of $u \zcirc v$, as defined in \eqref{eq:DefZhuProds}.  However, it turns out that one may obtain further vanishing relations by taking $n=-\frac{1}{2}$ (hence $u,v \in \zmsub{V}{1/2}$ or $\tzmsub{V}{1/2}$) and $k=1$ in \eqref{eq:GCR}.  The abstract version of these relations leads to the following extension of the product $\zcirc$ to both $\zmsub{V}{1/2}$ and $\tzmsub{V}{1/2}$:
\begin{equation}
u \zcirc v = \res{z=0} \sqbrac*{u(z) v \frac{(1+z)^{\wt u - 1/2}}{z}}.
\end{equation}
To impose the required vanishing and fix the non-associativity of $\zstar$, one introduces the ``subspaces of relations''
\begin{equation}
	\begin{aligned}
		\ozhu{\VOA{V}}&=\spn\set*{u\zcirc v\st u,v\in \zmsub{V}{0}} + \spn\set*{u\zcirc v\st u,v\in \zmsub{V}{1/2}}, \\
		\tozhu{\VOA{V}}&=\spn\set*{u\zcirc v\st u,v\in \tzmsub{V}{0}} + \spn\set*{u\zcirc v\st u,v\in \tzmsub{V}{1/2}}.	
	\end{aligned}
\end{equation}
These subspaces are in fact two-sided ideals of $\zmsub{V}{0}$ and $\tzmsub{V}{0}$, respectively, with respect to the product $\zstar$ \cite{Kacn1z94,DonTwi98}.
\begin{defn} \label{def:Zhu}
	The untwisted and $\tau$-twisted Zhu algebras of $\VOA{V}$ are the vector space quotients
	\begin{equation}
		\zhu{\VOA{V}}=\frac{\zmsub{V}{0}}{\ozhu{\VOA{V}}},\quad
    \tzhu{\VOA{V}}=\frac{\tzmsub{V}{0}}{\tozhu{\VOA{V}}},
	\end{equation}
	respectively, equipped with the product $\zstar$ defined in \eqref{eq:DefZhuProds}.
\end{defn}
\begin{rmk}
	In the literature, one normally finds the definition of $\zcirc$ extended further so that $v$ is not required to have the same moding as $u$.  This extension obviously has no interpretation in terms of the vanishing of zero modes, but allowing it leads to the non-integer moded elements being zero in the (twisted) Zhu algebra.  One can then extend $\zstar$ to all of $\VOA{V}$ by declaring that $u \zstar v$ is zero if either $u$ or $v$ is non-integer moded.  The utility of these extensions is not clear to us and they have the unfortunate consequence of obfuscating the relationship between (twisted) Zhu algebras and zero modes.
\end{rmk}

\begin{thm}[Kac and Wang {\cite[Thm.~1.1]{Kacn1z94}}; Dong, Li and Mason {\cite[Thm.~2.4(iii)]{DonTwi98}}] \label{thm:filtered}
  Both $\zhu{\VOA{V}}$ and $\tzhu{\VOA{V}}$ are unital associative algebras.  In each case, the unit is the image of the vacuum and the image of the conformal vector is central.  Moreover, both algebras are filtered, but not generally graded, by conformal weight.
\end{thm}

Let \(\Mod{M}\) be a (twisted) module over the \vosa{} \(\VOA{V}\) whose conformal
weights are bounded below, that is, there exists $r \in \RR$ such that the real part of every
\(L_0\)-eigenvalue on $\Mod{M}$ is at least $r$.
Then, the space \(\overline{\Mod{M}}\) of ground states of $\Mod{M}$ is non-zero.
Further, let \(\overline{\Mod{M}}\) be the subspace of \(\Mod{M}\) of vectors that are
annihilated by all positive modes of all fields in \(\VOA{V}\). For example,
the space of ground states of a Verma module \(\NSVer{h,c}\) over \(\svc{c}\) is spanned by its \hwv{} and all its singular vectors. The following results may also be found in \cite{Kacn1z94,DonTwi98}.

\begin{thm} \label{thm:zeromodehomom}
  Let \(\Mod{M}\) be an untwisted (twisted) module over a \vosa{} \(\VOA{V}\). Then, the subspace
  \(\overline{\Mod{M}}\) becomes an $\zhu{\VOA{V}}$-module ($\tzhu{\VOA{V}}$-module) on which the action of \([u] \in \zhu{\VOA{V}}\) (\([u] \in \tzhu{\VOA{V}}\)) is that of the zero mode \(u_0\) of its preimage \(u \in \zmsub{V}{0}\) (\(u \in \tzmsub{V}{0}\)).
\end{thm}

\begin{thm} \label{thm:zeromodeisom}
  There is a one-to-one correspondence between simple (twisted) Zhu algebra modules
  and simple (twisted) \(\VOA{V}\)-modules whose conformal weights are bounded below.  More precisely, the subspace \(\overline{\Mod{M}}\) of every such (twisted) $\VOA{V}$-module \(\Mod{M}\) is a simple (twisted) Zhu algebra module and every simple (twisted) Zhu algebra module
  \(\overline{\Mod{M}}\) can be induced to construct a unique simple (twisted) \(\VOA{V}\)-module \({\Mod{M}}\).
\end{thm}

\begin{rmk}
  A point to be emphasised is that the (twisted) Zhu algebra is, by
  construction (specifically, that presented above), canonically isomorphic
  to the algebra of zero modes acting on ground states of (twisted)
  modules.  In other words, \cref{thm:zeromodehomom} constructs an algebra
  homomorphism from the (twisted) Zhu to the algebra of zero modes (acting on ground states) and
  \cref{thm:zeromodeisom} implies that this homomorphism is indeed an
  isomorphism because every simple (twisted) Zhu algebra module can be 
  induced.
\end{rmk}

\flushleft

\end{document}